\documentclass[reqno, eucal]{amsart}

\usepackage[mathscr]{eucal} 
\usepackage[dvips]{graphicx}
\usepackage[dvips]{color}
\usepackage{amsmath,amsfonts,amssymb,amsthm,amscd}
\usepackage{epic}
\usepackage{eepic}
\usepackage{longtable}
\usepackage{array}
\usepackage{here}
\usepackage{cases}

\numberwithin{equation}{section}

\newtheorem{theorem}{Theorem}[section]
\newtheorem{lemma}[theorem]{Lemma}

\newtheorem{proposition}[theorem]{Proposition}
\newtheorem{corollary}[theorem]{Corollary}

\theoremstyle{definition}

\newtheorem{example}[theorem]{Example}

\newcommand{\Z}{{\mathbb Z}}
\newcommand{\R}{{\mathscr R}}
\newcommand{\C}{{\mathbb C}}

\newcommand{\s}{{\mathscr S}}

\begin{document}

\title[Multispecies TAZRP II]{Multispecies totally asymmetric zero range process:\\
II. Hat relation and tetrahedron equation}
\author{Atsuo Kuniba}
\email{atsuo@gokutan.c.u-tokyo.ac.jp}
\address{Institute of Physics, University of Tokyo, Komaba, Tokyo 153-8902, Japan}

\author{Shouya Maruyama}
\email{maruyama@gokutan.c.u-tokyo.ac.jp}
\address{Institute of Physics, University of Tokyo, Komaba, Tokyo 153-8902, Japan}

\author{Masato Okado}
\email{okado@sci.osaka-cu.ac.jp}
\address{Department of Mathematics, Osaka City University, 
3-3-138, Sugimoto, Sumiyoshi-ku, Osaka, 558-8585, Japan}


\maketitle

\begin{center}
{\it Dedicated to the memory of Professor Peter Petrovich Kulish}
\end{center}

\vspace{0.5cm}
\begin{center}{\bf Abstract}\end{center}
We consider a three-dimensional (3D)
lattice model associated with the intertwiner of 
the quantized coordinate ring $A_q(sl_3)$, and 
introduce a family of layer to layer transfer matrices on $m\times n$ square lattice.
By using the tetrahedron equation we derive their commutativity and bilinear relations 
mixing various boundary conditions.
At $q=0$ and $m=n$, they lead to a new proof of the steady state probability 
of the $n$-species totally asymmetric zero range process obtained recently 
by the authors, revealing the 3D integrability in the matrix product construction.

\vspace{0.4cm}

\section{Introduction}\label{sec:intro}

This is a continuation of the paper \cite{KMO3} hereafter called Part I, where 
an $n$-species totally asymmetric zero range process ($n$-TAZRP) was introduced.
It is a Markov process on one-dimensional periodic chain $\Z_L$ in which 
$n$-species of particles occupying the sites without an exclusion rule
hop to the adjacent sites only in one direction under the constraint that
smaller species ones have the priority.
For a general background concerning zero range processes, 
see for example \cite{EH, GSS, KL}.

In Part I it was shown that the $n$-TAZRP is the image of a projection from 
another stochastic system called {\em multiline process}. 
The latter has the uniform steady state measure, and by combining these facts, 
the steady state probability of the configuration  
$(\sigma_1, \ldots, \sigma_L)$ in the $n$-TAZRP has been obtained in the 
matrix product form
\begin{align}\label{ptx}
\mathbb{P}(\sigma_1, \ldots, \sigma_L)
=\mathrm{Tr}_{F^{\otimes n(n-1)/2}}(X_{\sigma_1}\cdots X_{\sigma_L}).
\end{align}
The operator $X_\sigma$ has the structure of a 
{\em corner transfer matrix} \cite{Bax} of a vertex model
whose Boltzmann weights are linear operators on 
the Fock space $F$.

Mathematically, the multiline process is 
a Markov process on the crystal \cite{Ka1} 
of a tensor product of the symmetric tensor representation
of the quantum affine algebra $U_q(\widehat{sl}_L)$.
The projection to the $n$-TAZRP is a composition of the 
combinatorial $R$ \cite{NY} which is the quantum $R$ matrix at $q=0$.

The theme of this paper 
is the 3D integrability of the $n$-TAZRP from the viewpoint of the 
{\em tetrahedron equation} \cite{Zam80}.
It is a 3D generalization of the Yang-Baxter equation \cite{Bax} whose 
relevance stands out for the multispecies case $n \ge 2$.
Our main result is a new proof of the matrix product formula (\ref{ptx})
by the so called {\em cancellation mechanism} or the {\em hat relation}
\begin{align}\label{htr}
\sum_{\gamma,\delta}
h^{\alpha,\beta}_{\gamma,\delta}X_\gamma X_\delta
= \hat{X}_\alpha X_\beta - X_\alpha \hat{X}_\beta,
\end{align}
where $h^{\alpha,\beta}_{\gamma,\delta}$ is the element of the 
local Markov matrix defined in (\ref{hrk2}) and (\ref{smta}).
Construction of such 
companion operators ${\hat X}_\alpha$ is a sufficient 
task to prove (\ref{ptx}) as is well known. See for example \cite{DEHP}.
In our setting of the $n$-TAZRP, 
the $X_{\alpha}$ and $\hat{X}_\alpha$ are linear operators 
on the space of ``internal degrees of freedom"  
$F^{\otimes n(n-1)/2}$ 
as depicted in the corner transfer matrix type diagram (\ref{Xz}).
As such, the relation (\ref{htr}) is highly nonlocal 
in the internal space with numerous summands, 
making the direct proof formidable.
To tame them most elegantly is the highlight of the paper.
Our strategy is to upgrade the statement 
by introducing $q$-melting\footnote{Deformation from 
the ``frozen point" $q=0$ rather than $q=1$.}, spectral parameters and embedding into a 
3D lattice model until the point where all the nonlocal commutation relations 
are integrated ultimately into a {\em single} local relation,
and that turns out to be the tetrahedron equation.

The solution of the tetrahedron equation relevant to our problem is  
the 3D $R$-operator obtained as 
the intertwiner of the quantized coordinate ring $A_q(sl_3)$ \cite{KV}.
It was also given in \cite{BS} and the two were identified in \cite{KO}.
It defines a vertex model on a cubic lattice 
whose edges are assigned with spin variables in $\Z_{\ge 0}$
and vertices with polynomial Boltzmann weights in $q$.
We introduce a family of transfer matrices for $m\times n$  layers
which depend on a spectral parameter and are  
labeled with boundary conditions.
By invoking the tetrahedron equation we establish  
their bilinear relations mixing various boundary labels.
Then the hat relation turns out to be a far-reaching 
consequence of a special case of $m=n$ and $q=0$.
We note that all these features are quite parallel with 
the $n$-species totally asymmetric simple exclusion process ($n$-TASEP) 
elucidated in \cite{KMO1, KMO2}.

After the present work, new discrete and continuous time integrable Markov process
on $n$-species of particles have been constructed in \cite{KMMO}.
The continuous time model therein contains the parameters $q$ and $\mu$.
It reproduces the $n$-TAZRP in this paper and Part I at $q=\mu=0$,
the $n$-species $q$-boson process 
in \cite{T} at $\mu=0$ and the model in \cite{T0} at $n=1$.

The paper is organized as follows.
In Section \ref{sec:tazrp}
we briefly recall the $n$-TAZRP and formulate the main result.
In Section \ref{sec:TE}
the 3D $R$-operator and its properties necessary in the later sections
are given.
In Section \ref{sec:lltm} the layer to layer transfer matrices are introduced  
and their bilinear relations are derived.
In Section \ref{sec:appli}  the specialization of the bilinear relation is analyzed and   
the proof of the hat relation is completed.
The paper is readable without consulting Part I, 
although the description of those overlapping parts is brief.

Throughout the paper we use the notations 
$(x)_+ = \max(x,0)$, 
$[i,j]=\{k \in \Z\mid i \le k \le j\}$, the $q$-Pochhammer symbol
$(z; q)_m = \prod_{j=1}^m(1-zq^{j-1})$, 
the $q$-factorial
$(q)_m = (q;q)_m$, the characteristic function 
$\theta(\mathrm{true})=1, 
\theta(\mathrm{false}) =0$ and the Kronecker delta
$\delta^{\alpha_1,\ldots, \alpha_m}_{\beta_1,\ldots, \beta_m} = 
\prod_{j=1}^m\theta(\alpha_j=\beta_j)$.

\section{Definitions and main result}\label{sec:tazrp}

\subsection{\mathversion{bold}$n$-TAZRP}
Let us quickly recapitulate the definition of the $n$-TAZRP.
A more detailed exposition is available in Part I.
Consider a periodic one-dimensional lattice $\Z_L$ of $L$ sites.
There are finitely many $n$-species of particles occupying the sites with 
no exclusion rule.
Thus the local state at a site is specified by the variable of the form
$\alpha=(\alpha_1,\ldots, \alpha_n) \in (\Z_{\ge 0})^n$
meaning that the number of species $a$ particles there 
is $\alpha_a$\footnote{ 
In Part I this multiplicity representation 
was denoted by $(\alpha^1,\ldots, \alpha^n)$ for distinction from 
the alternative multiset representation. 
In this paper we shall use the multiplicity representation only(!)
and ease the notation to $(\alpha_1,\ldots, \alpha_n)$.}.
We set $|\alpha| = \alpha_1+\cdots + \alpha_n$.
Thus a local state $\alpha=(2,0,1,3)$ for example signifies   
the site populated with the $|\alpha|=6$ particles $113444$. 

For two pairs of local states 
$({\gamma},{\delta})$ and $({\alpha}, {\beta})$ we define $>$ by
\begin{align}\label{smta}
(\gamma, \delta) > (\alpha, \beta) 
\overset{\text{def}}{\Longleftrightarrow}
(\alpha_j,\beta_j) = 
\begin{cases}
(\gamma_j+\delta_j,0) & 1 \le j < l,\\
(\gamma_l+d, \delta_l-d) & j=l,\\
(\gamma_j, \delta_j) & l<j \le n
\end{cases}
\;\;
\text{for some $l \in [1,n]$ and $d \in [1, \delta_l]$}.
\end{align}
It means that $(\alpha,\beta)$ is obtained from $(\gamma,\delta)$
by moving the smaller species
$\delta_1+\cdots+\delta_{l-1}+d$ 
particles from $\delta$ to $\gamma$.
Thus a particle can move only when it accompanies {\em all} the 
strictly smaller species ones than itself.
We let $(\gamma,\delta) \ge (\alpha,\beta)$ 
mean $(\gamma,\delta) > (\alpha,\beta)$ 
or $(\gamma,\delta)= (\alpha,\beta)$.
The definition of $>$ is more easily perceivable in terms of 
the multiset representation as in [Part I, eqs.(2.2), (2.4)].

The $n$-TAZRP is a stochastic dynamical system
in which neighboring pairs of local states
$(\sigma_i, \sigma_{i+1})=(\gamma,\delta)$ 
change into $(\alpha,\beta)$ such that 
$(\gamma,\delta) > (\alpha,\beta)$
with a uniform transition rate.
As the dynamics preserves the number of particles of each species,
the problem splits into {\em sectors} labeled with 
{\em multiplicity} ${\bf m}=(m_1,\ldots, m_n) \in (\Z_{\ge 0})^n$ of 
the species of particles:
\begin{align*}
S({\bf m}) = 
\{{\boldsymbol \sigma}=(\sigma_1,\ldots, \sigma_L)\mid
\sigma_i = (\sigma_{i,1}, \ldots, \sigma_{i,n})  \in (\Z_{\ge 0})^n,\;
\sum_{i=1}^L \sigma_{i,a}=m_a,\forall a \in [1,n]\}.
\end{align*}
Without loss of generality we shall exclusively consider the 
{\em basic sector} in which $m_a\ge 1$ for all $a \in [1,n]$.
Denote by ${\mathbb P}(\sigma_1,\ldots, \sigma_L; t)$
the probability of finding the system in the configuration
${\boldsymbol \sigma}=(\sigma_1,\ldots, \sigma_L)$ at time $t$, and let
$|P(t)\rangle
= \sum_{{\boldsymbol \sigma} \in S({\bf m})}
{\mathbb P}(\sigma_1,\ldots, \sigma_L; t)|\sigma_1,\ldots, \sigma_L\rangle$ 
be the vector representing the probability distribution in the basic sector 
$S({\bf m})$.
Our $n$-TAZRP is a Markov process governed by the master equation
$\frac{d}{dt}|P(t)\rangle
= H_{\mathrm{TAZRP}} |P(t)\rangle$,
where the Markov matrix has the form
\begin{align}\label{hrk2}
H_{\mathrm{TAZRP}}  = \sum_{i \in \Z_L} h_{i,i+1},\quad
h |\gamma,\delta\rangle =
\sum_{\alpha,\beta}h_{\gamma,\delta}^{\alpha,\beta}
|\alpha,\beta\rangle,\quad
h_{\gamma,\delta}^{\alpha,\beta}=
\begin{cases}
1 & \text{if }\; (\gamma,\delta) > (\alpha, \beta),\\
-|\beta| & \text{if } \;  (\gamma,\delta) = (\alpha, \beta),\\
0 & \text{otherwise}.
\end{cases}
\end{align}
Here $h_{i,i+1}$ is the local Markov matrix that  
acts as $h$ on the $i$-th and the $(i+1)$-th components and 
as the identity elsewhere.

Let 
$ |\bar{P}_L({\bf m})\rangle=\sum_{{\boldsymbol \sigma} \in S({\bf m})}
\mathbb{P}({\boldsymbol \sigma} ) |{\boldsymbol \sigma} \rangle$
be the {\em steady state} in the sector $S({\bf m})$.
The unnormalized 
$\mathbb{P}({\boldsymbol \sigma} )$ satisfying 
$\sum_{{\boldsymbol \sigma}\in S({\bf m})}
\mathbb{P}({\boldsymbol \sigma} )
= \prod_{a=1}^n \binom{L-1+\ell_a}{\ell_a}$ with
$\ell_a=\sum_{b\in [a,n]}m_b$ will be called the {\em steady state probability}
by abusing the terminology. 
The advantage for this is $\mathbb{P}({\boldsymbol \sigma} ) \in \Z_{\ge 1}$
as we will see.

\subsection{\mathversion{bold}Operators $X_\alpha$ and $\hat{X}_\alpha$}
\label{subsec:mr}
Let $F = \bigoplus_{m\ge 0} \C |m\rangle$ and 
$F^\ast = \bigoplus_{m\ge 0} \C \langle m |$ be the Fock space and its dual
with the bilinear pairing $\langle m | m' \rangle = \delta^{m}_{m'}$.
Let ${\bf a}^{\pm}, {\bf k}$ be the linear operators on them as
\begin{align}\label{iyo}
&{\bf a}^\pm|m\rangle = |m\pm 1\rangle, \quad
{\bf k}|m\rangle = \delta_m^0|m\rangle,
\quad
\langle m |{\bf a}^\pm = \langle m\mp1 |, \quad
\langle m |{\bf k} = \delta_m^0\langle m |
\end{align}
with $|\!-\!1\rangle = 0$ and 
$\langle\!-1| = 0$.
They obey the relations 
\begin{equation}\label{rna}
{\bf k} \ {\bf a}^{+}=0,\ \ \ \ \  \ {\bf a}^{-} \ {\bf k}=0,\ \ \ \ \ 
{\bf a}^+{\bf a}^-=1-{\bf k},\ \ \ \ \ {\bf a}^-{\bf a}^+=1
\end{equation}
corresponding to the $q=0$ case of the 
$q$-oscillator algebra $\mathscr{A}_q$ (\ref{tgm}).
In view of this, the algebra generated by ${\bf a}^{\pm}, {\bf k}$
with the relation (\ref{rna}) will be called the $0$-oscillator algebra 
and denoted by $\mathscr{A}_0$.
Note that ${\bf k}^2={\bf k}$.
The trace over $F$ is defined by 
$\mathrm{Tr}(X) = \sum_{m \ge 0}\langle m | X | m \rangle$
with $\langle m | m'\rangle = \delta^m_{m'}$.
It is easy to see that 
$\{({\bf a}^+)^f({\bf a}^-)^g| f,g  \in \Z_{\ge0}\}$ forms a basis 
of $\mathscr{A}_0$.

For $a,b,i,j \in \Z_{\ge 0}$ we introduce 
$\hat{R}^{a b}_{i j} \in \mathscr{A}_0$ and its diagram representation\footnote{
It is to be understood as the 2D projection of [Part I, eq.(5.6)].} as
\begin{equation} \label{vt}
\begin{picture}(200,38)(-40,-18)
\thinlines

\put(-51,-3){$\hat{R}^{a b}_{i j}=$}
\put(-10,0){\vector(1,0){20}}
\put(0,-10){\vector(0,1){20}}
\put(-17,-3.5){$i$}\put(12.5,-3.5){$a$}
\put(-2.4,13){$b$}\put(-2.3,-19){$j$}
\put(26,-4){$=\delta^{a+b}_{i+j}\theta(a \ge j)
({\bf a}^+)^j{\bf k}^{\theta(a>j)}({\bf a}^-)^b$,}
\thinlines
\end{picture}
\end{equation}
and regard it as the Boltzmann weight of 
the $\mathscr A_{0}$-valued vertex model on the 2D square lattice 
with edge variables $a,b,i,j \in \Z_{\ge0}$. 
The property that $\hat{R}^{a b}_{i j}=0$ 
unless $a+b=i+j$ will be referred to as {\em conservation law}.

The $n$-TAZRP operators $X_\alpha, \hat{X}_\alpha
\in \mathrm{End}(F^{\otimes n(n-1)/2})$ attached to the local state
$\alpha=(\alpha_1,\ldots,\alpha_n) \in (\Z_{\ge0})^n$ are given by
\begin{align}\label{XX}
X_\alpha = X_\alpha(z=1),\quad \hat{X}_\alpha  
= \frac{d}{dz} X_\alpha(z)|_{z=1},
\end{align}
where $X_\alpha(z) \in  \mathrm{End}(F^{\otimes n(n-1)/2})$\footnote{
In this paper we do not bother to write 
$F^{\otimes n(n-1)/2}[z,z^{-1}]$ etc.}
is defined by the following diagram\footnote{
We shall never abbreviate 
$X_\alpha(z)$ to $X_\alpha$, so the latter always means $X_\alpha(1)$.}:
\begin{equation}\label{Xz}
\begin{picture}(300,65)(58,6)

\put(105,0){
\put(10,0){
\put(-20,36){${\displaystyle X_\alpha(z) = 
\sum z^{a_1+a_2+\cdots+a_n}}$}
\put(120,0){
\put(0,-10){
\reflectbox{
\rotatebox[origin=c]{90}{
\put(20,52){$. . .$}
\put(-5,27){$.$}\put(-5,24){$.$}\put(-5,21){$.$}
\put(-7,48){\line(1,0){55}}
\put(-7,40){\line(1,0){47}}
\put(-7,32){\line(1,0){39}}
\put(-7,16){\line(1,0){23}}
\put(-7,8){\line(1,0){15}}
\put(-7,0){\line(1,0){7}}
\put(-9,-9.5){\put(29,25){$.$}\put(27,23){$.$}
\put(25,21){$.$}\put(31,27){$.$}\put(33,29){$.$}}
\put(48,48){\vector(0,1){8}}
\put(40,40){\vector(0,1){16}}
\put(32,32){\vector(0,1){24}}
\put(16,16){\vector(0,1){40}}
\put(8,8){\vector(0,1){48}}
\put(0,0){\vector(0,1){56}}
}}}
\put(48,63){\scriptsize${\alpha_n}$}
\put(10,53){\scriptsize{$\alpha_{n-1}+\alpha_n$}}
\put(-39,13){\scriptsize{$\alpha_1+\cdots+\alpha_n$}}
\put(73,60){\scriptsize{$a_1$}}
\put(73,52){\scriptsize{$a_2$}}
\put(73,40){\scriptsize{$\cdot$}}\put(73,35){\scriptsize{$\cdot$}}
\put(73,30){\scriptsize{$\cdot$}}\put(73,25){\scriptsize{$\cdot$}}
\put(73,20){\scriptsize{$\cdot$}}
\put(73,10){\scriptsize{$a_n$}}
}}

}

\end{picture}
\end{equation}
The sums are taken with respect to all the edge variables over $\Z_{\ge 0}$   
except the corners which are fixed depending on $ \alpha$.
The summand is the {\em tensor product} of 
$\hat{R}^{a b}_{i j} \in \mathscr{A}_0$ (\ref{vt}) attached to each vertex 
which acts on an independent copy of $F$ by (\ref{rna}).
The $X_\alpha(z)$ has the form of a {\em corner transfer matrix} \cite{Bax} of the 
$\mathscr{A}_0$-valued 2D vertex model.
It consists of infinitely many terms but 
the quantities we will deal with become always finite.
For ${\hat X}_{\alpha}$, the differentiation in (\ref{XX}) 
means an extra numerical coefficient 
$a_1+\cdots + a_n$.

\begin{example}\label{ex:n=2}
For $n=2$ the operator $X_{ \alpha}(z)$ is given by
\begin{equation*}
\begin{picture}(600,40)(-25,10)
\thinlines
\setlength\unitlength{0.26mm}
\put(50,0){
\put(-25,37){$X_{\alpha_1,\alpha_2}(z)= \ 
{\displaystyle \sum_{j\ge0}} \ z^{\alpha_1+\alpha_2+j}$}
\put(170,0){
\put(30,25){\line(0,1){35}}
\put(30,60){\vector(1,0){20}}
\put(10,40){\vector(1,0){40}}
\put(10,25){\line(0,1){15}}
\put(15,60){\scriptsize{$\alpha_2$}}
\put(-32,40){\scriptsize{$\alpha_1+\alpha_2$}}
\put(28,15){\scriptsize$j$}
\put(55,37){\scriptsize{$j+\alpha_1$}}

\put(100,37){$= \ z^{\alpha_1+\alpha_2}{\displaystyle \sum_{j\ge0}}z^j({\bf a}^+)^j{\bf k}^{\alpha_1}({\bf a}^-)^{\alpha_2}.$}
}
}
\thinlines
\end{picture}
\end{equation*}
\end{example}

\begin{example}\label{ex:n=3}
For $n=3$ the operator $\hat{X}_{ \alpha}$ is given by
\begin{equation*}
\begin{picture}(600,60)(-59,-17)
\thinlines
\setlength\unitlength{0.26mm}

\put(-55,24){$\hat{X}_{\alpha_1,\alpha_2,\alpha_3}= 
\ {\displaystyle \sum_{i,j,k}}(a_1+a_2+a_3)$}
\put(195,-30){
\reflectbox{
\rotatebox[origin=c]{90}{
\put(0,40){\line(1,0){50}} \put(50,40){\vector(0,1){23}}
\put(0,20){\line(1,0){30}} \put(30,20){\vector(0,1){43}}
\put(0,0){\line(1,0){10}}
\put(10,0){\vector(0,1){63}}
}}
\put(-43,80){\scriptsize{$\alpha_3$}}
\put(-90,60){\scriptsize{$\alpha_2+\alpha_3$}}
\put(-135,40){\scriptsize{$\alpha_1+\alpha_2+\alpha_3$}}
\put(-30,22){\scriptsize$j$}
\put(-50,22){\scriptsize$i$}
\put(-25,48){\scriptsize$k$}

\put(2,80){\scriptsize{$a_1$}}
\put(2,60){\scriptsize{$a_2$}}
\put(2,40){\scriptsize{$a_3$}}

}

\put(5,-27){$= \ {\displaystyle \sum_{i,j,k}}
(|{\alpha}|+i+j)({\bf a}^+)^j
{\bf k}^{\alpha_1+i-k}({\bf a}^-)^k{\otimes}
({\bf a}^+)^k{\bf k}^{\alpha_2}({\bf a}^-)^{\alpha_3}
{\otimes}({\bf a}^+)^i{\bf k}^{\alpha_1}({\bf a}^-)^{\alpha_2+\alpha_3}$.}

\end{picture}
\end{equation*}
\\ 
The sums are taken over $\{i,j,k|i,j\ge0,0\le k \le \alpha_1+i\}$. 
The edge variables corresponding to $a_1,a_2,a_3$ in (\ref{Xz})
have been determined by using the conservation law as
$a_1=\alpha_3$, $a_2=\alpha_2+k$, $a_3=\alpha_1+i+j-k$.
Here and in what follows, 
the components of the tensor product will always be ordered 
so that they correspond, from left to right, 
to the vertices (if exist) at 
$(1,1)$, $(2,1)$, $(1,2)$, $(3,1)$, $(2,2)$, $(1,3), \ldots$, 
where $(i,j)$ is the intersection of the $i$-th horizontal line from the bottom
and the $j$-th vertical line from the right.
\end{example}

\subsection{Main result}
Now we state the main result of this paper. 
\begin{theorem}[Hat relation]\label{main}
For any $\alpha, \beta \in (\Z_{\ge 0})^n$, the operators (\ref{XX}) satisfy
\begin{align*}
\sum_{ \gamma, \delta}h^{\alpha,\beta}_{\gamma,\delta}
X_ { \gamma}X_{ \delta}
= \hat{X}_\alpha X_\beta - X_\alpha \hat{X}_\beta,
\end{align*}
where the sum extends over all the $n$-$\mathrm{TAZRP}$ local states
$\gamma, \delta \in (\Z_{\ge 0})^n$. 
\end{theorem}
The left hand side is a finite sum due to (\ref{hrk2}).  
The proof will be achieved in the end of Section \ref{sec:appli},
where a Baxterization, i.e., 
spectral parameter dependent generalization (\ref{ngjj}) 
is obtained.  

\begin{corollary}\label{ss}
The steady state probability of the $n$-$\mathrm{TAZRP}$ is expressed as
\begin{align*}
{\mathbb P}( \sigma_1,\ldots, \sigma_L) 
= \mathrm{Tr}(X_{ \sigma_1}\cdots X_{ \sigma_L}),
\end{align*}
\end{corollary}
where the trace is taken over $F^{\otimes n(n-1)/2}$.
\begin{proof}
The convergence of the trace is guaranteed by 
the argument after [Part I, Th.5.8].
Thus we are left to show $H_{\mathrm{TAZRP}}
\sum_{{\boldsymbol \sigma} = (\sigma_1,\ldots, \sigma_L) \in S({\bf m})}
\mathrm{Tr}(X_{ \sigma_1}\cdots X_{ \sigma_L})
|{\boldsymbol \sigma}\rangle=0$.
The left hand side is equal to  
\begin{align*}
&\sum_{i \in \Z_L}\sum_{\boldsymbol \sigma \in S({\bf m})}
\sum_{\sigma'_i, \sigma'_{i+1}}
\mathrm{Tr}(\cdots X_{\sigma_i}X_{\sigma_{i+1}}\cdots )
h^{\sigma'_i, \sigma'_{i+1}}_{\sigma_i, \sigma_{i+1}}
|\ldots, \sigma'_i, \sigma'_{i+1},\ldots\rangle\\
&= \sum_{\boldsymbol \sigma \in S({\bf m})}\sum_{i \in \Z_L}
\mathrm{Tr}(\cdots (hXX)_{\sigma_i, \sigma_{i+1}}\cdots )
|\ldots, \sigma_i, \sigma_{i+1},\ldots\rangle,
\end{align*}
where $(hXX)_{\alpha, \beta}$ denotes the left hand side of the hat relation 
in Theorem \ref{main}.
After replacing it with the right hand side, 
the sum $\sum_{i \in \Z_L}$ vanishes 
thanks to the cyclicity of the trace.
\end{proof}

\begin{example}
For the $2$-TAZRP on length 3 chain in the sector $S(2,1)$, we have
\begin{align*}
&{\mathbb P}(10,10,01) = \mathrm{Tr}(X_{1,0}X_{1,0}X_{0,1})
=\sum_{i,j,k\ge0}\mathrm{Tr}\bigl(({\bf a}^+)^i{\bf k}({\bf a}^+)^j{\bf k}({\bf a}^+)^k{\bf a}^-\bigl)=1,\\
&{\mathbb P}(00,20,01) = \mathrm{Tr}(X_{0,0}X_{2,0}X_{0,1})
=\sum_{i,j,k\ge0}\mathrm{Tr}\bigl(({\bf a}^+)^i({\bf a}^+)^j{\bf k}({\bf a}^+)^k{\bf a}^-\bigl)=2,
\end{align*}
reproducing the coefficients of 
$|1,1,2\rangle$ and $|\emptyset,11,2\rangle$ 
in $|\xi_3(2,1)\rangle$ in [Part I, Ex. 2.1], respectively.
\end{example}

In the following sections, 
we will identify the operator $X_\alpha(z)$ (\ref{Xz}) as
a piece of the layer to layer transfer matrix in a 3D system as in (\ref{S=X}).

\section{3D $R$-operators and tetrahedron equation}\label{sec:TE}

Here we introduce the $q$-version of the objects in the previous section 
such as $F, F^\ast, {\bf a}^\pm, {\bf k}$.
For simplicity we use the same notation for them.

\subsection{\mathversion{bold}$q$-oscillators}
Let $q$ be a generic parameter. 
Let $F = \bigoplus_{m \ge 0}\C(q)|m\rangle$ and 
$F^\ast = \bigoplus_{m \ge 0}\C(q)\langle m|$ be the Fock space and its dual
with the bilinear pairing $\langle m | m'\rangle = (q^2)_m\delta^m_{m'}$.
Let ${\bf a}^+,{\bf a}^-,{\bf k}$  be the operators acting on them as 
\begin{equation}\label{irte}
\begin{split}
{\bf a}^+|m\rangle &=|m+1\rangle, 
\quad {\bf a}^-|m\rangle=(1-q^{2m})|m-1\rangle, \quad 
{\bf k}|m\rangle=q^m|m\rangle,\\
\langle m |{\bf a}^- &=\langle m+1|,\quad
\langle m |{\bf a}^+=(1-q^{2m})\langle m-1|,\quad
\langle m |{\bf k} = q^m\langle m |.
\end{split}
\end{equation} 
They satisfy
\begin{equation}\label{tgm}
{\bf k} \, {\bf a}^{\pm}=q^{\pm1}{\bf a}^{\pm}\,{\bf k},
\ \ \ \ \ {\bf a}^+\,{\bf a}^-=1-{\bf k}^2,
\ \ \ \ \ {\bf a}^-\,{\bf a}^+=1-q^2{\bf k}^2.
\end{equation}
The pairing fulfills $\langle m| (X|m'\rangle) = (\langle m|X)|m'\rangle$.
The algebra generated by ${\bf a}^\pm, {\bf k}$ with these relations will be called
the $q$-oscillator algebra ${\mathscr A}_q$. 
It reduces to $\mathscr{A}_0$ in Section \ref{subsec:mr} at $q=0$.

\subsection{\mathversion{bold}3D $R$-operator with spectral parameter}
Define the operators
$\hat{\R}^{ab}_{ij}(z), \hat{\s}^{ab}_{ij}(z) \in \mathscr{A}_q$ 
depending on the {\em spectral parameter} $z$ by\footnote{
Likewise for $F$, we do not bother to write 
$\mathscr{A}_q[z,z^{-1}]$ etc. in this paper.}
\begin{equation}\label{rsdef}
\hat{\R}^{ab}_{ij}(z)=
\hat{\s}^{ba}_{ji}(z^{-1}) = 
\delta^{a+b}_{i+j}\,z^{j-b}\!\!\sum_{\lambda+\mu=b}(-1)^\lambda 
q^{\lambda+\mu^2-ib}\binom{i}{\mu}_{\!q^2}
\binom{j}{\lambda}_{\!q^2}({\bf a}^-)^{\mu}({\bf a}^+)^{j-\lambda}
{\bf k}^{i+\lambda-\mu}
\end{equation}
for $a,b,i,j \in \Z_{\ge 0}$.
The sum extends over $\lambda \in [0,j], \mu\in [0,i]$ such that 
$\lambda+\mu=b$, and 
$\binom{m}{j}_{\!q}=\frac{(q)_m}{(q)_j(q)_{m-j}}$ is the 
$q$-binomial coefficient. 
We depict them as 3D vertices as follows:
\begin{equation}\label{3Dpic}
\begin{picture}(220,50)(-310,-22)
\thinlines
\put(-280,5){
\put(-75,-7){$\hat{\R}^{ab}_{ij}(z)\;=$}

\rotatebox{20}{
{\linethickness{0.2mm}
\put(13,-4){\color{blue}\vector(-1,0){40}}}}

\put(-5,-18){\vector(0,1){32}}
\put(-20,0){\vector(3,-1){33}}

\put(-26,0){$\scriptstyle{i}$}
\put(-6,17){$\scriptstyle{b}$}
\put(-6,-26){$\scriptstyle{j}$}
\put(15,-14){$\scriptstyle{a}$}
}

\put(-100,5){
\put(-75,-7){$\hat{\s}^{ab}_{ij}(z)\;=$}

\rotatebox{20}{
{\linethickness{0.2mm}
\put(13,-4){\color{green}\vector(-1,0){40}}}}

\put(-5,-18){\vector(0,1){32}}
\put(-20,0){\vector(3,-1){33}}

\put(-26,0){$\scriptstyle{i}$}
\put(-6,17){$\scriptstyle{b}$}
\put(-6,-26){$\scriptstyle{j}$}
\put(15,-14){$\scriptstyle{a}$}
}
\end{picture}
\thinlines
\end{equation}
The blue and green arrows are to be understood as the operators in
$\mathscr{A}_q$ which act either on $F$ or $F^\ast$. 
The two colors are used to distinguish 
$\hat{\R}^{ab}_{ij}(z)$ and $\hat{\s}^{ab}_{ij}(z)$. 
Although they are simply related by the interchange 
$(a,i,z) \leftrightarrow (b,j,z^{-1})$, 
keeping the both will turn out to be useful in our working below.
The $z$-dependence not exhibited in the diagrams will be specified 
whenever necessary.

By a direct calculation we find
\begin{alignat}{2}
&\hat{\R}^{ab}_{ij}(z) |k\rangle 
= z^{j-b}\sum_c \R^{abc}_{ijk} |c\rangle, &
&\langle c|\hat{\R}^{ab}_{ij}(z) = z^{j-b}\sum_k\frac{(q^2)_c}{(q^2)_k}
\R^{abc}_{ijk}\langle k|, 
\label{arkr}\\
&\hat{\s}^{ab}_{ij}(z) |k\rangle 
= z^{j-b}\sum_c\mathscr{R}^{bac}_{jik} |c\rangle,
&\qquad
&\langle c|\hat{\s}^{ab}_{ij}(z) = z^{j-b}\sum_k\frac{(q^2)_c}{(q^2)_k}
\mathscr{R}^{bac}_{jik}\langle k|.
\label{csz}
\end{alignat}
Here and in what follows, the sum like $\sum_c$ means $\sum_{c \in \Z_{\ge 0}}$
unless otherwise stated. 
The coefficient $\mathscr{R}^{abc}_{ijk}$ is given by
\begin{equation}\label{def:R}
\R^{abc}_{ijk}=\delta^{a+b}_{i+j}\delta^{b+c}_{j+k}
\sum_{\lambda+\mu=b}(-1)^{\lambda}
q^{i(c-j)+(k+1)\lambda+\mu(\mu-k)}
\frac{(q^2)_{c+\mu}}{(q^2)_c}\binom{i}{\mu}_{\!q^2}
\binom{j}{\lambda}_{\!q^2} \in \Z[q]
\end{equation}
with the sum taken under the same condition as (\ref{rsdef}).
Based on $\hat{\R}^{ab}_{ij}(z), \hat{\s}^{ab}_{ij}(z)$ we introduce
\begin{align*}
\R(z)|i,j,k\rangle &=\sum_{a,b}|a\rangle \otimes |b\rangle
\otimes \hat{\R}^{ab}_{ij}(z)|k\rangle
= \sum_{a,b,c}
z^{j-b}\R^{abc}_{ijk}|a,b,c\rangle, 
\\
\s(z)|i,j,k\rangle &= \sum_{a,b} |a\rangle
\otimes |b\rangle \otimes \hat{\s}^{ab}_{ij}(z)|k\rangle
= \sum_{a,b,c}
z^{j-b}\R^{bac}_{jik}|a,b,c\rangle,
\end{align*}
where $|i,j,k\rangle=|i\rangle \otimes |j\rangle \otimes |k\rangle$ etc. 
We call $\R(z), \s(z) \in \mathrm{End}(F^{\otimes 3})$ 
the {\em 3D $R$-operators}.
Naturally they are depicted by the corresponding diagrams in (\ref{3Dpic})
with no specification of the values $a,b,i,j$.
Set $\langle i,j,k| = \langle i| \otimes  \langle j| \otimes  \langle k|$
similarly. 
Then the right action of $\R(z)$ on $F^\ast$ is defined as
$\langle a,b,c|\R(z)
= \sum_{i,j,k} z^{j-b}\R^{ijk}_{abc}\langle i,j,k|$
to be compatible with 
$(\langle a,b,c|\R(z))|i,j,k\rangle = 
\langle a,b,c| (\R(z) |i,j,k\rangle)$ and 
the rightmost relation in (\ref{R:pro}).
We denote the {\em constant} 3D $R$-operator by\footnote{
We shall never abbreviate $\R(z)$ to $\R$, so the latter always means  $\R(1)$.}
\begin{align*}
\R=\R(1).
\end{align*}
The full 3D $R$-operators depending on
the spectral parameter are recovered from it by
\begin{equation}\label{hRh}
\R(z)_{123}=z^{-{\bf h}_2}\R_{123}z^{{\bf h}_2}
=z^{{\bf h}_1}\R_{123}z^{-{\bf h}_1}, \quad
\s(z)_{123}=z^{-{\bf h}_2}\R_{213}z^{{\bf h}_2}
=z^{{\bf h}_1}\R_{213}z^{-{\bf h}_1}
 \end{equation}
in terms of the ${\bf h} $ acting on $F$ and $F^\ast$ by
${\bf h}|m\rangle=m|m\rangle$ and $\langle m|{\bf h}=\langle m|m$.
The indices $1,2,3$ specify the copies of $F$ in  
$\overset{1}{F}\otimes\overset{2}{F}\otimes\overset{3}{F}$
on which these operators operators act. 
 
\subsection{Tetrahedron equation}

The constant 3D $R$-operator $\R$ has the following properties:
\begin{align}
&[\R_{123},x^{{\bf h}_1}(xy)^{{\bf h}_2}y^{{\bf h}_3}]=0, 
\label{comm}\\
&\R=\R^{-1}, \quad 
\R_{123}=\R_{321}, \quad 
\R^{abc}_{ijk}=\frac{(q^2)_i(q^2)_j(q^2)_k}{(q^2)_a(q^2)_b(q^2)_c}\R^{ijk}_{abc}. 
\label{R:pro}
\end{align}
The first one, where $x, y$  are generic parameters, 
follows straightforwardly from (\ref{def:R}). 
For (\ref{R:pro}) see \cite{KO} where the definitions of $\mathscr{R}$ and 
$\mathscr{R}^{abc}_{ijk}$ are identical with this paper.
The second relation means $\R^{abc}_{ijk} = \R^{cba}_{kji}$.
The most significant property of the 3D $R$-operators is the following.
\begin{theorem}[Tetrahedron equation with spectral parameter]\label{th:te}
Set $z_{ij}=z_i/z_j$ where $z_1,\ldots, z_4$ are generic.
As an operator on $F^{\otimes 6}$, the following equality holds:
\begin{equation}\label{TEz}
\s(z_{12})_{126}\s(z_{34})_{346}\R(z_{13})_{135}\R(z_{24})_{245}
=\R(z_{24})_{245}\R(z_{13})_{135}\s(z_{34})_{346}\s(z_{12})_{126}.
\end{equation}
\end{theorem} 
\begin{proof}
By substituting (\ref{hRh}) into (\ref{TEz}) and applying (\ref{comm}), 
one finds that the similarity transformation
$z_{13}^{-{\bf h}_1}z_{23}^{-{\bf h}_2}z_{34}^{{\bf h}_4}(\ref{TEz})
z_{13}^{{\bf h}_1}z_{23}^{{\bf h}_2}z_{34}^{-{\bf h}_4}$
removes the $z$-dependence completely
reducing it to the constant tetrahedron equation
$\R_{216}\R_{436}\R_{135}\R_{245}
=\R_{245}\R_{135}\R_{436}\R_{216}$.
Due to the second property in (\ref{R:pro}), 
this coincides with \cite[eq.(2.40)]{KO} with the indices changed as 
$1 \leftrightarrow 5$,\ $3 \to 2$,\ $2 \to 4$,\ $4 \to 3$.
\end{proof}
Graphically the tetrahedron equation (\ref{TEz}) is expressed as follows:
\begin{equation}\label{kiip}
\begin{picture}(400,87)(-80,-1)
\setlength{\unitlength}{0.32mm}
\thinlines

\put(59,50){$5$}
\rotatebox{35}{\put(70,14){\color{blue}\vector(-1,0){70}}}

\put(7,2){$3$}
\put(10,12){\line(0,1){38}}\qbezier(10,50)(10,51.5)(10,53)
\put(10,53){\vector(2,3){26}}

\put(33,30){$4$}
\put(41,34){\line(0,1){28}}\put(40,65){\vector(-1,1){25}}
\qbezier(41,62)(40.5,63.5)(40,65)

\put(16,42){$2$}
\put(25,45){\line(1,0){30}}\put(58,44){\vector(3,-2){27}}
\qbezier(55,45)(56.5,44.5)(58,44)

\put(-14,22){$1$}
\put(-5,24){\line(1,0){37}}
\qbezier(32,24)(34.5,24.5)(37,25)
\put(37,25){\vector(3,1){48}}

\put(1,77){$6$}
{\color{green}\qbezier(10,79)(67,80)(73.5,19)}
\put(74,18){\color{green}\vector(1,-4){1}}

\put(115,40){$=$}
\put(190,31){

\put(-48,23){$1$}
\put(-41,26){\line(3,1){51}}\put(13,43){\vector(1,0){43}}
\qbezier(10,43)(11.5,43)(13,43)

\put(-46,39){$2$}
\put(-14,23){\line(-3,2){23}}\put(-12,22){\vector(1,0){38}}
\qbezier(-14,23)(-13,22.5)(-12,22)

\put(12,-40){$3$}
\put(42,24){\vector(0,1){34}}
\qbezier(42,24)(41.5,21)(39.8,18)
{\rotatebox[origin=l]{-10}{\put(31.5,24){\line(-1,-3){17}}}}

\put(23,-30){$4$}
\put(-8,1){
\put(10,0){\vector(0,1){38}}\put(11,-2){\line(1,-1){20}}
\qbezier(10,0)(10.5,-1)(11,-2)}

\put(44,50){$5$}
\put(-20,-2){
\rotatebox{35}{\put(70,14){\color{blue}\vector(-1,0){70}}}}

\put(-39,50){$6$}
\put(-8,9){
{\color{green}\qbezier(-32,37)(-32,-23)(35,-22)}
\put(34.5,-22){\color{green}\vector(1,0){1}}}

}

\end{picture}
\thinlines
\end{equation}
Every arrow carries a Fock space $F$.
The spectral parameters in (\ref{TEz}) 
can consistently be encoded by assigning 
$z_1,\ldots, z_4$ to the black arrows $1,\ldots, 4$, respectively.

The constant 3D $R$-operator $\R$ with a formula like (\ref{def:R}) 
was obtained in \cite{KV} (albeit with misprint) 
as an intertwiner of irreducible representations 
of the quantized coordinate ring $A_q(sl_3)$. 
By the construction it satisfies the constant tetrahedron equation.
The $\R$ was also given in \cite{BS,BMS} in a different gauge 
from a quantum geometry consideration.
The two were identified in \cite[eq.(2.29)]{KO}.
The operator (\ref{rsdef}) was introduced in \cite{Ku} for $z=1$.
Theorem \ref{th:te} is a slight generalization of the constant tetrahedron equation
by the spectral parameters.
Although their dependence is of simply removable nature,
the relation (\ref{TEz}) turns out to be essential 
in our analysis of the $n$-TAZRP.

\subsection{\mathversion{bold}Eigenvectors of the 3D $R$-operator}
We introduce the following definitions\footnote{
The $|\chi(z)\rangle$ and $\langle \chi(z)|$ should actually 
be considered in a completion of $F$ and $F^\ast$.} :
\begin{align}
&|\chi(z)\rangle=\sum_{m \ge 0}\chi_m(z)|m \rangle \in F,
\quad
\langle \chi(z)|=\sum_{m \ge 0} \chi_m(z)\langle m| \in F^\ast,
\label{xvec}\\
&\chi_m(z)=\chi_m \,z^m,\quad
\chi'_m(z) = \chi'_m \,z^m,\quad
\chi_m = \frac{1}{(q)_m},\quad
\chi'_m = \frac{(q^2)_m}{(q)_m}.
\label{nska}
\end{align}

\begin{lemma} $\mathrm{(}$\cite[Pro. 4.1]{KS}$\mathrm{)}$ 
\label{le:ks}
The constant 3D $R$-operator has the eigenvectors as
\begin{align*}
\R\bigl(
|\chi(x)\rangle {\otimes} |\chi(xy)\rangle {\otimes} |\chi(y)\rangle
\bigr)
&=|\chi(x)\rangle {\otimes} |\chi(xy)\rangle {\otimes} |\chi(y)\rangle,
\\
\bigl(\langle \chi(x) | \otimes
\langle \chi(xy) | \otimes 
\langle \chi(y) | \bigr)\R 
&= \langle \chi(x) | \otimes
\langle \chi(xy) | \otimes 
\langle \chi(y) |.
\end{align*}
\end{lemma}

We adapt it to the form applicable to our analysis in the next section.

\begin{lemma}\label{aaR}
\begin{equation}\label{aaR>}
\sum_{i,j}\chi_i(\lambda)\chi_j(\mu)
\hat{\s}^{ab}_{ij}(z)|\chi(\textstyle \frac{\lambda}{\mu z})\rangle
=\chi_a(\lambda)\chi_b(\mu)
|\chi(\textstyle \frac{\lambda}{\mu z})\rangle,
\end{equation}
\begin{equation}\label{aa<R}
\sum_{a,b}\chi'_a(\lambda)\chi'_b(\mu)\langle 
\chi(\textstyle \frac{\lambda z}{\mu})|
\hat{\s}^{ab}_{ij}(z)
=\chi'_i(\lambda)\chi'_j(\mu)\langle
\chi(\textstyle \frac{\lambda z}{\mu})|.
\end{equation}
\end{lemma}
\begin{proof}
From $\chi_i(\lambda)\chi_j(\mu)\hat{\s}^{ab}_{ij}(z)
=\lambda^a\mu^b\chi_i\chi_j\hat{\s}^{ab}_{ij}(\frac{\mu z}{\lambda})$ 
and $\chi'_a(\lambda)\chi'_b(\mu)\hat{\s}^{ab}_{ij}(z)
=\lambda^i\mu^j\chi'_a\chi'_b\hat{\s}^{ab}_{ij}(\frac{\lambda z}{\mu})$, 
the proof reduces to the case $\lambda=\mu=1$. 
By the first relation in Lemma \ref{le:ks} with $x=y=1$ we know
$\chi_a\chi_b\chi_c=\sum_{i,j,k}\chi_i\chi_j\chi_k \R^{abc}_{ijk}$.
Applying $\sum_c z^{-c}|c\rangle$ to this 
and noting that $\R^{abc}_{ijk}=0$ unless 
$(a+b,b+c)=(i+j,j+k)$,  we get
$\chi_a\chi_b|\chi(z^{-1})\rangle=
\sum_{i,j}\chi_i\chi_j\sum_k\chi_kz^{-k}
z^{i-a}\sum_c\R^{abc}_{ijk}|c\rangle
=\sum_{i,j}\chi_i\chi_j\sum_k\chi_kz^{-k}\hat{\s}^{ba}_{ji}(z)|k\rangle
=\sum_{i,j}\chi_i\chi_j\hat{\s}^{ba}_{ji}(z)|\chi(z^{-1})\rangle$.
The relation (\ref{aa<R}) can be shown similarly.
\end{proof}

\section{Layer to Layer transfer matrices}\label{sec:lltm}

For any $N \in \Z_{\ge 1}$ and 
array ${\bf i}=(i_1,\ldots, i_N) \in (\Z_{\ge 0})^N$,
we will use the notation
\begin{equation}\label{chic}
\begin{split}
\chi_{\bf i}&= \chi_{i_1}\cdots \chi_{i_N}, \qquad\qquad\quad\;\;\;
\chi'_{\bf i}= \chi'_{i_1}\cdots \chi'_{i_N},\\
\chi_{\bf i}(z) &= \chi_{i_1}(z)\cdots \chi_{i_N}(z), \qquad
\chi'_{\bf i}(z) = \chi'_{i_1}(z)\cdots \chi'_{i_N}(z),
\end{split}
\end{equation}
where the right hand sides are defined in (\ref{nska}).
Recall also that $|{\bf i}| = i_1+\cdots + i_N$
as already used in (\ref{hrk2}).

\subsection{\mathversion{bold}Definition of layer to layer transfer matrices}
Fix positive integers $n$ and $m$. 
Associated with the arrays 
${\bf b}=(b_1,{\ldots}, b_n), 
{\bf j}=(j_1,{\ldots}, j_n)  \in (\Z_{\ge 0})^n$
and 
${\bf a}=(a_1,{\ldots},a_m), 
{\bf i}=(i_1,{\ldots},i_m) \in (\Z_{\ge 0})^m$,
we introduce $T(z)^{{\bf a},{\bf b}}_{{\bf i}, {\bf j}} \in 
\mathscr{A}_q^{\otimes mn}$ by
\begin{equation}
\begin{picture}(300,88)(-95,-17)
\put(-84,26){$T(z)^{{\bf a},{\bf b}}_{{\bf i},{\bf j}}\,= 
{\displaystyle \sum}$}
\put(-8,48){\vector(1,0){56}}
\put(-8,40){\vector(1,0){56}}
\put(-8,32){\vector(1,0){56}}
\put(-8,8){\vector(1,0){56}}
\put(-8,0){\vector(1,0){56}}

\put(40,-5){\vector(0,1){60}}
\put(32,-5){\vector(0,1){60}}
\put(8,-5){\vector(0,1){60}}
\put(0,-5){\vector(0,1){60}}

\put(-5,60){${\scriptstyle b_1}$}
\put(5,60){${\scriptstyle b_2}$}
\put(20,60){$. . .$}
\put(38,60){${\scriptstyle b_n}$}

\put(-5,-12){${\scriptstyle j_1}$}
\put(5,-12){${\scriptstyle j_2}$}
\put(20,-12){$. . .$}
\put(38,-12){${\scriptstyle j_n}$}

\put(50,48){${\scriptstyle a_1}$}
\put(50,38){${\scriptstyle a_2}$}
\put(50,14){$\vdots$}
\put(50,-2){${\scriptstyle a_m}$}

\put(-17,48){${\scriptstyle i_1}$}
\put(-17,38){${\scriptstyle i_2}$}
\put(-17,14){$\vdots$}
\put(-17,-2){${\scriptstyle i_m}$}

\put(125,28){
\put(-10,0){\vector(1,0){20}}
\put(0,-10){\vector(0,1){20}}
\put(-17,-3.5){$i$}\put(12.5,-3.5){$a$}
\put(-2.4,13){$b$}\put(-2.3,-19){$j$}
\put(30,-4){$=\, \hat{\R}^{ab}_{ij}(z)$\, in (\ref{rsdef}).}
}
\end{picture}
\label{asks}
\end{equation}
The sum here is taken 
with respect to all the internal edges over $\Z_{\ge 0}$.
It is the ``partition function" of a 2D 
vertex model with fixed boundary condition.
The ``Boltzmann weight" of a configuration is the {\em tensor product}  of 
$\hat{\R}^{ab}_{ij}(z) \in \mathscr{A}_q$ (\ref{rsdef}) attached to each vertex\footnote{
It is {\em not}  (\ref{vt}) although the same diagram is used for simplicity.
They are identified at $q=0$ in Lemma \ref{le:ysgk}.}.
It is naturally regarded as an element in 
$\mathrm{End}(F^{\otimes mn})$ by (\ref{irte}).
By using $T(z)^{{\bf a}, {\bf b}}_{{\bf i}, {\bf j}}$, 
we define the main object of our study:
\begin{equation} \label{bT}
\mathbb{T}(z)^{\bf b}_{\bf i}
= \chi'_{\bf b}\,\chi_{\bf i}
\sum_{{\bf a}, \,{\bf j}} \chi'_{\bf a}\,\chi_{\bf j}\,
T(z)^{{\bf a}, {\bf b}}_{{\bf i}, {\bf j}},
\end{equation}
where the sum ranges over all 
${\bf a} = (a_1,\ldots, a_m) \in (\Z_{\ge 0})^m$ and 
${\bf j} = (j_1, \ldots, j_n) \in (\Z_{\ge 0})^n$. 
This is the partition function of a 2D 
vertex model with the NW-fixed and the SE-free boundary condition
with the extra  Boltzmann weight 
$\chi'_{\bf b} \chi_{\bf i} \chi'_{\bf a} \chi_{\bf j}$ attached to
the boundary edges.
The $\mathbb{T}(z)^{\bf b}_{\bf i} \in \mathscr{A}_q^{\otimes mn}$ 
is also regarded as an element in 
$\mathrm{End}(F^{\otimes mn})$.

The diagram for $\hat{\R}^{ab}_{ij}(z)$ in (\ref{asks}) is to be understood as
the 2D projection of the actual 3D vertex in (\ref{3Dpic}).
When emphasizing this aspect,  we employ more 3D looking diagrams.
For instance (\ref{bT}) with $(m,n)=(3,4)$ is depicted as
the layer in the cubic lattice as
\begin{equation}
\begin{picture}(200,80)(-98,-10)
\setlength{\unitlength}{0.55mm}

\put(-80,22){$\mathbb{T}(z)^{\bf b}_{\bf i}
= \chi'_{\bf b}\,\chi_{\bf i}
\displaystyle \sum_{{\bf a}, \,{\bf j}} \chi'_{\bf a}\,\chi_{\bf j}$}
\put(-11,29){$i_1$}
\put(-11,19){$i_2$}
\put(-11,9){$i_3$}

\put(46,24){$a_1$}
\put(46,14){$a_2$}
\put(46,4){$a_3$}

\put(3,39){$b_1$}\put(13,38){$b_2$}\put(23,37){$b_3$}\put(33,36){$b_4$}

\put(0,-43){
\put(3,39){$j_1$}\put(13,38){$j_2$}\put(23,37){$j_3$}\put(33,36){$j_4$}
}

\put(-4,0){
\put(0,30){\rotatebox{-5}{\vector(1,0){47}}}
\put(0,20){\rotatebox{-5}{\vector(1,0){47}}}
\put(0,10){\rotatebox{-5}{\vector(1,0){47}}}

\multiput(-1,2)(10,-1){4}{
\put(10,0){\vector(0,1){35}}}

\multiput(10,10)(10,-1){4}{
\multiput(0,0)(0,10){3}{
\put(2.4,1.6){\color{blue}\vector(-3,-2){9}}
}}
}
\end{picture}
\label{gkkn}
\end{equation}
All the vertices here are penetrated from back to face 
by the blue arrows carrying the independent copies of the 
Fock space as (\ref{3Dpic}). 
The sum is taken not only for ${\bf a}=(a_1,a_2,a_3)$ and 
${\bf j}=(j_1,j_2,j_3, j_4)$ but also for 
all the internal edges.
In this way, 
one may either view the $\mathbb{T}(z)^{\bf b}_{\bf i}$ as a
partition function of 
the $\mathscr{A}_q$-valued 2D vertex model,
or as a layer to layer transfer matrix of the 3D lattice model as in (\ref{gkkn}).
In the latter picture, the Boltzmann weight assigned with the 
vertex  (\ref{3Dpic})  is $z^{j-b}\R^{abc}_{ijk}$ (\ref{def:R}) when 
the blue arrow goes from $k$ to $c$.
By the definition the $z$-dependence of each summand is a 
simple power $z^{|{\bf j}|-|{\bf b}|}$.

\begin{example}\label{ex:S^00_0}
For $(m,n)=(1,2)$ the definition (\ref{bT}) reads
\begin{equation*}
\!\!\!\!\!\!\!\!\!\!\!\!\!\!
\!\!\!\!\!\!\!\!\!\!\!\!\!\!\!\!\!\!
\begin{picture}(140,35)(0,0)
\thinlines
\mbox{
\begin{picture}(70,18)(-6,12)
\put(30,15){
\put(-87,12){${\mathbb T}(z)^{b_1,b_2}_{i}
=\chi_i\chi'_{b_1,b_2}{\displaystyle \sum_{j_1,j_2}}$}

\put(30,0){
\put(6,14){\vector(0,1){10}}
\put(18,14){\vector(0,1){10}}
\put(-2,14){\line(1,0){8}}
\put(6,14){\line(1,0){12}}
\put(18,14){\vector(1,0){11}}
\put(6,7){\line(0,1){7}}
\put(18,7){\line(0,1){7}}

\put(4,27){\scriptsize{$b_1$}}
\put(15,27){\scriptsize{$b_2$}}
\put(-7,12){\scriptsize{$i$}}
\put(4,0){\scriptsize{$j_1$}}
\put(15,0){\scriptsize{$j_2$}}
\put(37,12){$\chi'_{i+j_1+j_2-b_1-b_2}\chi_{j_1,j_2}
$}
}
}
\end{picture}
}
\end{picture}
\end{equation*}
The sum 
is taken over $\{j_1,j_2|j_1 \ge (b_1-i)_+, \ j_2 \ge (b_1+b_2-i-j_1)_+\}$.
For example one has
\begin{align*}
&{\mathbb T}(z)^{0,0}_0=\sum_{j_1, j_2 \ge 0}
z^{j_1+j_2}\chi'_{j_1+j_2}\chi_{j_1,j_2}({\bf a}^+)^{j_2}{\bf k}^{j_1}{\otimes}({\bf a}^+)^{j_1}, \\
&{\mathbb T}(z)^{0,0}_1=\frac{1}{1-q}\sum_{j_1, j_2 \ge 0}
z^{j_1+j_2}\chi'_{j_1+j_2+1}\chi_{j_1,j_2}({\bf a}^+)^{j_2}{\bf k}^{j_1+1}{\otimes}({\bf a}^+)^{j_1}{\bf k}, \\
&{\mathbb T}(z)^{1,0}_0=
-z^{-1}(1+q)q\sum_{j_1 \ge 1, j_2 \ge 0}
z^{j_1+j_2}\frac{1-q^{2j_1}}{1-q^2}\chi'_{j_1+j_2-1}\chi_{i_1,j_2}({\bf a}^+)^{j_2}{\bf k}^{j_1-1}{\otimes}({\bf a}^+)^{j_1-1}{\bf k}.
\end{align*}
\end{example} 

\subsection{Bilinear relations}
\begin{proposition}\label{pr:stt}
For any $x,x', y, y'$ and 
arrays ${\bf a}, {\bf a}', {\bf i}, {\bf i}' \in (\Z_{\ge 0})^m$
and 
${\bf b}, {\bf b}', {\bf j}, {\bf j}' \in (\Z_{\ge 0})^n$,
the following equality holds as an operator on 
$F \otimes F^{\otimes mn}$:
\begin{align}\label{eq:ltol}
&\sum_{{\bf a}'',{\bf a}''',{\bf b}'',{\bf b}'''}
\Bigl(\hat{\s}^{a_ma'_m}_{a''_ma'''_m}(\textstyle \frac{x}{x'}){\cdots}\hat{\s}^{a_1a'_1}_{a''_1a'''_1}(\textstyle \frac{x}{x'})\Bigr)
\Bigl(\hat{\s}^{b_nb'_n}_{b''_nb'''_n}(\textstyle \frac{y}{y'}){\cdots}
\hat{\s}^{b_1b'_1}_{b''_1b'''_1}(\textstyle \frac{y}{y'})\Bigr)
T(\textstyle \frac{x}{y})^{{\bf a}''{\bf b}''}_{{\bf i} \ {\bf j}}T(\textstyle \frac{x'}{y'})^{{\bf a}'''{\bf b}'''}_{{\bf i}' \ {\bf j}'}  
\nonumber \\
& =  \sum_{{\bf i}'',{\bf i}''',{\bf j}'',{\bf j}'''}
T(\textstyle \frac{x'}{y'})^{{\bf a}'{\bf b}'}_{{\bf i}'''{\bf j}'''}
T(\textstyle \frac{x}{y})^{{\bf a} \ {\bf b}}_{{\bf i}''{\bf j}''}
\Bigl(\hat{\s}^{j''_nj'''_n}_{j_nj'_n}(\textstyle \frac{y}{y'}){\cdots}
\hat{\s}^{j''_1j'''_1}_{j_1j'_1}(\textstyle \frac{y}{y'})\Bigr)
\Bigl(\hat{\s}^{i''_mi'''_m}_{i_mi'_m}(\textstyle \frac{x}{x'}){\cdots}{\hat \s}^{i''_1i'''_1}_{i_1i'_1}(\textstyle \frac{x}{x'})\Bigr),
\end{align}
where the sums are taken over the arrays 
${\bf a}'', {\bf a}''', {\bf i}'', {\bf i}''' \in (\Z_{\ge 0})^m$,
${\bf b}'', {\bf b}''', {\bf j}'', {\bf j}''' \in (\Z_{\ge 0})^n$.
Each array is specified by the components as 
${\bf a}''=(a_1'',{\ldots}, a_m'') \in (\Z_{\ge 0})^m$, etc.
\end{proposition}

\begin{proof}
The relation (\ref{eq:ltol}) is depicted as follows:
\begin{equation} \label{fig}
\begin{picture}(400,140)(-43,-10)
\setlength{\unitlength}{0.6mm}
\thicklines
{\color{green}\drawline(0,60)(60,60)
\qbezier(60,60)(63,60)(63,57)\put(63,57){\vector(0,-1){60}}}
\thinlines
\put(-25,32){$\displaystyle{
\sum_{\substack{{\bf a}'',{\bf a}''',\\{\bf b}'',{\bf b}'''}}}$}
\put(-10,25){
\qbezier(67,18)(68,18)(69,18)\put(69,18){\vector(1,-1){10}}
\put(14,18){\line(1,0){5}}\put(21,18){\line(1,0){26}}\put(49,18){\line(1,0){18}}
\put(8,17){$i'_1$}\put(64,20.9){$a'''_1$}\put(84.5,16){$a_1$}
\put(10,10){\line(1,0){50}}
\qbezier(60,10)(62,10.2)(63,10.3)
\put(63,10.3){\vector(3,1){20}}
\put(4,8.2){$i_1$}\put(63,6){$a''_1$}\put(80,5){$a'_1$}
\put(12,-3){$\vdots$}
}

\put(-10,-3){
\qbezier(67,18)(68,18)(69,18)\put(69,18){\vector(1,-1){10}}
\put(14,18){\line(1,0){5}}\put(21,18){\line(1,0){26}}\put(49,18){\line(1,0){18}}
\put(6,17){$i'_m$}\put(64,20.9){$a'''_m$}\put(84.5,16){$a_m$}
\put(10,10){\line(1,0){50}}
\qbezier(60,10)(62,10.2)(63,10.3)
\put(63,10.3){\vector(3,1){20}}
\put(2.5,8){$i_m$}\put(63,6){$a''_m$}\put(80,5){$a'_m$}
}

\put(-10,-7){
\put(30.5,63){\vector(-1,1){9.5}}\qbezier(30.5,63)(30.9,62.7)(31,62)
\put(31,10){\line(0,1){3}}\put(31,15){\line(0,1){26}}\put(31,43){\line(0,1){19}}
\put(15,74){$b'_1$}\put(30,78){$b_1$}
\put(33,61){$b'''_1$}\put(30,4){$j'_1$}
\put(20.4,57){\vector(2,3){12}}\qbezier(20.4,57)(20.3,56.8)(20,56)
\put(20,6){\line(0,1){50}}
\put(14,57){$b''_1$}\put(18,0){$j_1$}
\put(0.2,7){\rotatebox{36}{\put(41,25){\color{blue}\vector(-1,0){22}}}}
\put(0.2,-21){\rotatebox{36}{\put(41,25){\color{blue}\vector(-1,0){22}}}}
\put(35,53){$\cdots$}
}

\put(18,-7){
\put(30.5,63){\vector(-1,1){9.5}}\qbezier(30.5,63)(30.9,62.7)(31,62)
\put(31,10){\line(0,1){3}}\put(31,15){\line(0,1){26}}\put(31,43){\line(0,1){19}}
\put(15,74){$b'_n$}\put(30,78){$b_n$}
\put(33,61){$b'''_n$}\put(30,4){$j'_n$}
\put(20.4,57){\vector(2,3){12}}\qbezier(20.4,57)(20.3,56.8)(20,56)
\put(20,6){\line(0,1){50}}
\put(14,57){$b''_n$}\put(18,0){$j_n$}
\put(0.2,7){\rotatebox{36}{\put(41,25){\color{blue}\vector(-1,0){22}}}}
\put(0.2,-21){\rotatebox{36}{\put(41,25){\color{blue}\vector(-1,0){22}}}}
}

\put(-10,0){
\put(93,32){$=
\displaystyle{\sum_{\substack{{\bf i}'',{\bf i}''',\\{\bf j}'',{\bf j}'''}}}$}
\put(-10,0){
\thicklines
\put(139,5){
{\color{green}\drawline(0,60)(0,3)
\qbezier(0,3)(0,0)(3,0)\put(3,0){\vector(1,0){60}}}
\thinlines

\put(-1,32){
\drawline(-8,12)(10,18)\qbezier(10,18)(11,18)(12.6,18)
\put(15,18){\line(1,0){26}}\put(43,18){\vector(1,0){24}}
\put(4,20){$i''_1$}\put(69,17){$a_1$}
\drawline(-6,20)(6,12)\qbezier(6,12)(8,10)(10,10)
\put(10,10){\vector(1,0){50}}
\put(2.2,5){$i'''_1$}\put(62,9){$a'_1$}
\put(-11.5,21){$i'_1$}\put(-13.5,10){$i_1$}
\put(-5,0){$\vdots$}
}

\put(-1,4){
\drawline(-8,12)(10,18)\qbezier(10,18)(11,18)(12.6,18)
\put(15,18){\line(1,0){26}}\put(43,18){\vector(1,0){24}}
\put(4,20){$i''_m$}\put(69,17){$a_m$}
\drawline(-6,20)(6,12)\qbezier(6,12)(8,10)(10,10)
\put(10,10){\vector(1,0){50}}
\put(1.8,5.5){$i'''_m$}\put(62,9){$a'_m$}
\put(-13.5,20){$i'_m$}\put(-15,10){$i_m$}
}

\put(-7,0){
\drawline(21,-10)(31,8)  \qbezier(31,8)(31,9) (31,10)
\put(31,10){\line(0,1){3}}\put(31,15){\line(0,1){26}}\put(31,43){\vector(0,1){19}}\put(29,64){$b_1$}\put(32.2,7){$j''_1$}\put(32.3,-7.5){$j'_1$}
\drawline(32,-5)(21,4)  \qbezier(21,4)(21,4) (20,6)
\put(20,6){\vector(0,1){50}}
\put(18,58){$b'_1$}\put(12,2.7){$j'''_1$}\put(16.3,-12){$j_1$}
\put(0.2,7){\rotatebox{36}{\put(41,25){\color{blue}\vector(-1,0){22}}}}
\put(0.2,-21){\rotatebox{36}{\put(41,25){\color{blue}\vector(-1,0){22}}}}
\put(36,55){$\cdots$}
}

\put(21,0){
\drawline(21,-10)(31,8)  \qbezier(31,8)(31,9) (31,10)
\put(31,10){\line(0,1){3}}\put(31,15){\line(0,1){26}}\put(31,43){\vector(0,1){19}}
\put(29,64){$b_n$}\put(32.2,7){$j''_n$}\put(32.3,-7.5){$j'_n$}
\drawline(32,-5)(21,4)  \qbezier(21,4)(21,4) (20,6)
\put(20,6){\vector(0,1){50}}
\put(18,58){$b'_n$}\put(12,3){$j'''_n$}\put(16.1,-12){$j_n$}
\put(0.2,7){\rotatebox{36}{\put(41,25){\color{blue}\vector(-1,0){22}}}}
\put(0.2,-21){\rotatebox{36}{\put(41,25){\color{blue}\vector(-1,0){22}}}}
}

}
}
}
\end{picture}
\thinlines
\end{equation}
The $T(z)^{{\bf a},{\bf b}}_{{\bf i}, {\bf j}}$'s act on the
$F^{\otimes mn}$ on the blue arrows 
and the $\hat{\s}^{ab}_{ij}(z)$'s 
do on the single $F$ on the green arrow.
On the left hand side, starting from the top right corner, 
one can apply the tetrahedron equation (\ref{kiip}) successively
to push the green arrow down through all the blue arrows.
It converts the left hand side into the right hand side.
To check the fitness of the spectral parameters with (\ref{TEz}),
assign the four groups of black arrows labeled by the external edges
${\bf a}, {\bf a}', {\bf b}$ and ${\bf b}'$ with  
$x,x',y$ and $y'$, respectively.
\end{proof}

Note that $\hat{\s}^{ab}_{ij}(z)=0$ unless $a+b=i+j$ due to (\ref{rsdef}).
Therefore the sums in (\ref{eq:ltol}) are subject to 
${\bf a}''+{\bf a}''' = {\bf a}+{\bf a}', 
{\bf b}''+{\bf b}''' = {\bf b}+{\bf b}'$ on the left hand side and
${\bf i}''+{\bf i}''' = {\bf i}+{\bf i}', 
{\bf j}''+{\bf j}''' = {\bf j}+{\bf j}'$ on the right hand side,
hence are finite.
Now we are ready to prove the main result in this section.

\begin{theorem}[Bilinear relation of layer to layer transfer matrix]\label{Th:main}
For any ${\bf s} \in (\Z_{\ge 0})^n$ and ${\bf r} \in (\Z_{\ge 0})^m$, 
the following relation holds
as an operator on $F^{\otimes mn}$:
\begin{equation*}
\sum_{\substack{{\bf b},{\bf b}',{\bf i},{\bf i}'\\{\bf b}+{\bf b}'={\bf s},\,{\bf i}+{\bf i}'={\bf r}}}
x^{|{\bf b}|+|{\bf i}|}y^{|{\bf b}'|+|{\bf i}'|}\,
{\mathbb T}(x)^{\bf b}_{\bf i}\,{\mathbb T}(y)^{{\bf b}'}_{{\bf i}'}=(x \leftrightarrow y),
\end{equation*}
where the sum extends over
${\bf b}, {\bf b}' \in (\Z_{\ge 0})^n$ and 
${\bf i}, {\bf i}' \in (\Z_{\ge 0})^m$ under the specified conditions.
\end{theorem}
\begin{proof}
In (\ref{eq:ltol})  let us set ${\bf b}+{\bf b}'={\bf s}, {\bf i}+{\bf i}'={\bf r}$ and 
evaluate 
\begin{align*}
\sum_{\substack{{\bf a},{\bf a}',{\bf j},{\bf j}',{\bf b},{\bf b}',{\bf i},{\bf i}'\\
{\bf b}+{\bf b}'={\bf s},\,{\bf i}+{\bf i}'={\bf r}}}
\chi'_{\bf a}\chi'_{{\bf a}'}\chi_{\bf j}\chi_{{\bf j}'}\chi'_{\bf b}(\textstyle \frac{x}{y})\chi'_{{\bf b}'}(\textstyle \frac{x'}{y'})\chi_{\bf i}(\textstyle \frac{x}{y})
\chi_{{\bf i}'}(\textstyle \frac{x'}{y'})
\langle \chi(\textstyle \frac{x}{x'})|(\cdots)|\chi(\textstyle \frac{y'}{y})\rangle
\end{align*}
in each side. 
The matrix element 
$\langle \chi(\textstyle \frac{x}{x'})|(\cdots)|\chi(\textstyle \frac{y'}{y})\rangle$ 
is calculated along the Fock space on the green arrow in (\ref{fig}). 
For the left hand side,  (\ref{aa<R}) can be applied to show 
\begin{align*}
&\sum_{\substack{{\bf a},{\bf a}',{\bf b},{\bf b}'\\{\bf b}+{\bf b}'={\bf s}}}
\chi'_{\bf a}\chi'_{{\bf a}'}\chi'_{\bf b}(\textstyle \frac{x}{y})
\chi'_{{\bf b}'}(\textstyle \frac{x'}{y'})
\langle \chi(\textstyle \frac{x}{x'})|
\hat{\s}^{{a_m}{a'_m}}_{{a''_m}{a'''_m}}(\textstyle \frac{x}{x'}) \cdots 
\hat{\s}^{{a_1}{a'_1}}_{{a''_1}{a'''_1}}(\textstyle \frac{x}{x'}) 
\hat{\s}^{b_n  b'_n}_{b''_nb'''_n}(\textstyle \frac{y}{y'}) \cdots 
\hat{\s}^{b_1  b'_1}_{b''_1b'''_1}(\textstyle \frac{y}{y'}) \\
&\qquad =\chi'_{{\bf a}''}\chi'_{{\bf a}'''}
\chi'_{{\bf b}''}(\textstyle \frac{x}{y})
\chi'_{{\bf b}'''}(\textstyle \frac{x'}{y'})\langle \chi(\textstyle \frac{x}{x'})|.
\end{align*}
Similarly in the right hand side, (\ref{aaR>}) provides the simplification
\begin{align*}
&\sum_{\substack{{\bf j},{\bf j}',{\bf i},{\bf i}'\\{\bf i}+{\bf i}'={\bf r}}}
\chi_{\bf j}\chi_{{\bf j}'}\chi_{\bf i}(\textstyle \frac{x}{y})
\chi_{{\bf i}'}(\textstyle \frac{x'}{y'})
\hat{\s}^{j''_nj'''_n}_{j_nj'_n}(\textstyle \frac{y}{y'}) \cdots 
\hat{\s}^{j''_1j'''_1}_{j_1j'_1}(\textstyle \frac{y}{y'}) 
\hat{\s}^{i''_mi'''_m}_{i_m  i'_m}(\textstyle \frac{x}{x'})\cdots 
\hat{\s}^{i''_1i'''_1}_{i_1  i'_1}(\textstyle \frac{x}{x'})|\chi(\textstyle \frac{y'}{y})\rangle \\
&\qquad =\chi_{{\bf j}''}\chi_{{\bf j}'''}
\chi_{{\bf i}''}(\textstyle \frac{x}{y})
\chi_{{\bf i}'''}(\textstyle \frac{x'}{y'})|\chi(\textstyle \frac{y'}{y})\rangle.
\end{align*}
From these relations, we obtain
\begin{align*}
&\sum_{\substack{{\bf a}'',{\bf a}''',{\bf j},{\bf j}',{\bf b}'',{\bf b}''',{\bf i},{\bf i}'\\
{\bf b}''+{\bf b}'''={\bf s}, \,{\bf i}+{\bf i}'={\bf r}}}
\chi'_{{\bf a}''}\chi'_{{\bf a}'''}
\chi_{\bf j}\chi_{{\bf j}'}
\chi'_{{\bf b}''}(\textstyle \frac{x}{y})
\chi'_{{\bf b}'''}(\textstyle \frac{x'}{y'})
\chi_{\bf i}(\textstyle \frac{x}{y})
\chi_{{\bf i}'}(\textstyle \frac{x'}{y'})
T(\textstyle \frac{x}{y})^{{\bf a}''{\bf b}''}_{{\bf i} \ {\bf j}}
T(\textstyle \frac{x'}{y'})^{{\bf a}'''{\bf b}'''}_{{\bf i}'{\bf j}'}  \\
&=\sum_{\substack{{\bf a},{\bf a}',{\bf j}'',{\bf j}''',{\bf b},{\bf b}',{\bf i}'',{\bf i}'''\\
{\bf b}+{\bf b}'={\bf s}, \,{\bf i}''+{\bf i}'''={\bf r}}}
\chi'_{\bf a}\chi'_{{\bf a}'}\chi_{{\bf j}''}\chi_{{\bf j}'''}
\chi'_{\bf b}(\textstyle \frac{x}{y})
\chi'_{{\bf b}'}(\textstyle \frac{x'}{y'})
\chi_{{\bf i}''}(\textstyle \frac{x}{y})
\chi_{{\bf i}'''}(\textstyle \frac{x'}{y'})
T(\textstyle \frac{x'}{y'})^{{\bf a}'{\bf b}'}_{{\bf i}'''{\bf j}'''}
T(\textstyle \frac{x}{y})^{{\bf a} \ {\bf b}}_{{\bf i}''{\bf j}''}.
\end{align*}
Here we have divided the both sides 
by $\langle \chi(\textstyle \frac{x}{x'})|\chi(\textstyle \frac{y'}{y})\rangle
=\sum_{m \ge 0}\frac{(q^2)_m}{(q)_m^2}(\frac{xy'}{x'y})^m \neq 0$. 
In view of the definition (\ref{bT}) this is equivalent to 
\begin{align*}
&{\displaystyle \sum_{\substack{{\bf b},{\bf b}',{\bf i},{\bf i}'\\
{\bf b}+{\bf b}'={\bf s},\, {\bf i}+{\bf i}'={\bf r}}}}\!\!\!
(\textstyle \frac{x}{y})^{|\bf b|+|\bf i|}(\textstyle \frac{x'}{y'})^{|{\bf b}'|+|{\bf i}'|}
{\mathbb T}(\textstyle \frac{x}{y})^{\bf b}_{\bf i}{\mathbb T}
(\textstyle \frac{x'}{y'})^{{\bf b}'}_{{\bf i}'} 
=\!\!\!\!\!\!{\displaystyle \sum_{\substack{{\bf b},{\bf b}',{\bf i},{\bf i}'\\
{\bf b}+{\bf b}'={\bf s},\, {\bf i}+{\bf i}'={\bf r}}}}\!\!\!
(\textstyle \frac{x'}{y'})^{|{\bf b}'|+|{\bf i}'|}
(\textstyle \frac{x}{y})^{|\bf b|+|\bf i|}
{\mathbb T}(\textstyle \frac{x'}{y'})^{{\bf b}'}_{{\bf i}'}
{\mathbb T}({\textstyle \frac{x}{y}})^{\bf b}_{\bf i}.
\end{align*} 
\end{proof}

Let us isolate the special case 
${\bf s}=(0,{\ldots},0,s) \in (\Z_{\ge 0})^n$, 
${\bf r}=(0,{\ldots},0,r) \in (\Z_{\ge 0})^m$, 
which will be utilized in the next section.

\begin{corollary}
For any $s,r \in \Z_{\ge 0}$, the following equality is valid:
\begin{equation}\label{xySS}
\sum_{\substack{b,b',i,i'\\b+b'=s, \, i+i'=r}}x^{b+i}y^{b'+i'}
{\mathbb T}(x)^{0,{\ldots},0,b}_{0,{\ldots},0,i}\,
{\mathbb T}(y)^{0,{\ldots},0,b'}_{0,{\ldots},0,i'}=(x \leftrightarrow y).
\end{equation} 
In particular we have a commuting family of layer to layer transfer matrices:
\begin{equation}\label{eq:comm}
[{\mathbb T}(x)^{\bf 0}_{\bf 0},{\mathbb T}(y)^{\bf 0}_{\bf 0}]=0,
\end{equation}
where ${\bf 0}$ denotes $(0,\ldots, 0)$ either in $(\Z_{\ge 0})^m$ or $(\Z_{\ge 0})^n$.
\end{corollary}   

\begin{example}\label{ex:s^00_0-2}
From Example \ref{ex:S^00_0} we have
\begin{align*}
&{\mathbb T}(x)^{0,0}_0 \,{\mathbb T}(y)^{0,0}_0
=\sum_{j_1,j_2,j'_1,j'_2}x^{j_1+j_2}y^{j'_1+j'_2}
\chi'_{j_1+j_2}\chi'_{j'_1+j'_2}\chi_{j_1,j_2}\chi_{j'_1,j'_2}
({\bf a}^+)^{j_2}{\bf k}^{j_1}({\bf a}^+)^{j'_2}{\bf k}^{j'_1}{\otimes}({\bf a}^+)^{j_1+j'_1} \\
&=\sum_{r,s,t}x^ry^s\chi'_r\chi'_s
({\bf a}^+)^{r+s-t}{\bf k}^{t}{\otimes}({\bf a}^+)^{t} f_{r,s,t}(q),
\end{align*}
where ${\bf k}\,{\bf a}^+ = q{\bf a}^+{\bf k}$ is used.
The coefficient $ f_{r,s,t}(q)$, which is $0$ unless $r+s\ge t$, is given by
\begin{equation*}
f_{r,s,t}(q)=
\sum_{\substack{j_1,j_2,j'_1,j'_2 \\j_1+j_2=r,\, j'_1+j'_2=s, \, j_1+j'_1=t}}
\frac{q^{j_1j'_2}}{(q)_{j_1}(q)_{j_2}(q)_{j'_1}(q)_{j'_2}}.
\end{equation*}
Therefore (\ref{eq:comm}) implies $f_{r,s,t}(q)=f_{s,r,t}(q)$.
In fact it is easily confirmed by 
comparing the coefficients of $z^t$ 
on the both sides of $(-z;q)_r(-zq^r;q)_s=(-z;q)_s(-zq^s;q)_r$.
\end{example}

\begin{example}\label{ex:s^00_0-3}
From Example \ref{ex:S^00_0} we have
\begin{align*}
&x{\mathbb T}(x)^{0,0}_1{\mathbb T}(y)^{0,0}_0+y{\mathbb T}(x)^{0,0}_0{\mathbb T}(y)^{0,0}_1 \\
&=\frac{x}{1-q} \sum_{j_1,j_2,j'_1,j'_2}x^{j_1+j_2}y^{j'_1+j'_2}\chi'_{j_1+j_2+1}\chi'_{j'_1+j'_2}\chi_{j_1,j_2,j'_1,j'_2}({\bf a}^+)^{j_2}{\bf k}^{j_1+1}({\bf a}^+)^{j'_2}{\bf k}^{j'_1}\otimes({\bf a}^+)^{j_1}{\bf k}({\bf a}^+)^{j'_1} \\
&+\frac{y}{1-q} \sum_{j_1,j_2,j'_1,j'_2}x^{j_1+j_2}y^{j'_1+j'_2}\chi'_{j_1+j_2}\chi'_{j'_1+j'_2+1}\chi_{j_1,j_2,j'_1,j'_2}({\bf a}^+)^{j_2}{\bf k}^{j_1}({\bf a}^+)^{j'_2}{\bf k}^{j'_1+1}\otimes({\bf a}^+)^{j_1+j'_1}{\bf k} \\
&=\frac{1}{1-q}\sum_{r,s,t}x^ry^s\chi'_r\chi'_s
\sum_{\substack{j_1,j_2,j'_1,j'_2\\j_1+j_2+1=r,j'_1+j'_2=s,j_1+j'_1=t}}q^{j_1j'_2+s}\chi_{j_1,j_2,j'_1,j'_2}({\bf a}^+)^{r+s-t-1}{\bf k}^{t+1}\otimes({\bf a}^+)^{t}{\bf k} \\
&+\frac{1}{1-q}\sum_{r,s,t}x^ry^s\chi'_r\chi'_s
\sum_{\substack{j_1,j_2,j'_1,j'_2\\j_1+j_2=r,j'_1+j'_2+1=s,j_1+j'_1=t}}q^{j_1j'_2}\chi_{j_1,j_2,j'_1,j'_2}({\bf a}^+)^{r+s-t-1}{\bf k}^{t+1}\otimes
({\bf a}^+)^{t}{\bf k} \\
&=\frac{1}{1-q}\sum_{r,s,t}
x^ry^s\chi'_r\chi'_s({\bf a}^+)^{r+s-t-1}{\bf k}^{t+1}
\otimes({\bf a}^+)^{t}{\bf k} \ (q^sf_{r-1,s,t}(q)+f_{r,s-1,t}(q)).
\end{align*}
Therefore (\ref{xySS}) with $s=0$ and $r=1$ implies
$q^sf_{r-1,s,t}(q)+f_{r,s-1,t}(q)=(r \leftrightarrow s)$.
From $f_{r,s,t}(q)=f_{s,r,t}(q)$, it is equivalent to 
$(1-q^s)f_{r-1,s,t}=(r \leftrightarrow s)$. 
In fact it is derived from 
$(-z;q)_{r-1}(-zq^{r-1};q)_s=(-z;q)_{s-1}(-zq^{s-1};q)_r$.
\end{example}

\section{Application to $n$-TAZRP}\label{sec:appli}

Here we specialize the results in Section \ref{sec:TE} and \ref{sec:lltm}
to $q=0$ and $m=n$.
All the objects remain well-defined.
In particular the oscillators ${\bf a}^+, {\bf a}^-, {\bf k}$ 
in this section mean those in $\mathscr{A}_0$ 
defined by (\ref{iyo}) and (\ref{rna}).
We write 
$\alpha_{\ge j} = \alpha_j+\alpha_{j+1}+\cdots + \alpha_n$ 
for an array $\alpha=(\alpha_1,\ldots, \alpha_n) \in 
(\Z_{\ge 0})^n$.

\subsection{\mathversion{bold}$R$-operators at $q=0$}\label{ss:q=0}

Since $\R^{abc}_{ijk}$ (\ref{def:R}) is a polynomial in $q$,
one can safely set
\begin{align}\label{tann}
R^{abc}_{ijk}:=\R^{abc}_{ijk}|_{q=0}
=\delta^a_{j+(i-k)_+}\delta^b_{{\rm min}(i,k)}\delta^c_{j+(k-i)_+},
\end{align}
where the latter equality is due to 
\cite[eq.(2.32)]{KO}
and $\R^{-1}=\R$ (\ref{R:pro}).

\begin{lemma}\label{le:ysgk}
${\hat \R}^{ab}_{ij}(z)|_{q=0} = z^{j-b}{\hat R}^{ab}_{ij}$.
\end{lemma}
\begin{proof}
Recall that $\hat{\R}^{ab}_{ij}(z)$ is specified by (\ref{arkr}) with (\ref{def:R}),
and $\hat{R}^{ab}_{ij}$ is defined in (\ref{vt}).
As the both sides are operators in $\mathscr{A}_0$,
it suffices to check the equality under the evaluation  
$\langle c |(\cdots)|k\rangle$ for arbitrary $c,k \in \Z_{\ge 0}$. 
From (\ref{arkr}),  we see that 
$\langle c |z^{b-j}{\hat \R}^{ab}_{ij}(z)|_{q=0}|k\rangle=
\R^{abc}_{ijk}|_{q=0}=R^{abc}_{ijk}$. 
On the other hand $\langle c |{\hat R}^{ab}_{ij}|k\rangle$ has been calculated in 
[Part I, Lem. 5.4] and the result agrees 
with the rightmost expression in (\ref{tann}).
\end{proof}

\subsection{\mathversion{bold}Layer to Layer transfer matrix at $q=0$}

In the reminder of this section 
we let $\mathbb{T}(z)^{\bf b}_{\bf i}$ denote
the layer to layer transfer matrix (\ref{bT})$|_{m=n}$ {\em specialized at $q=0$}.
By Lemma \ref{le:ysgk} and (\ref{chic}), we find 
\begin{equation}
\begin{picture}(300,88)(-95,-17)
\put(-110,26){$\mathbb{T}(z)^{\bf b}_{\bf i}\,= 
{\displaystyle \sum_{{\bf a}, {\bf j}}}\,\;z^{|{\bf j}|-|{\bf b}|}$}
\put(-7,40){\vector(1,0){56}}
\put(-7,31){\vector(1,0){56}}
\put(-7,9){\vector(1,0){56}}
\put(-7,0){\vector(1,0){56}}

\put(40,-5){\vector(0,1){54}}
\put(31,-5){\vector(0,1){54}}
\put(9,-5){\vector(0,1){54}}
\put(0,-5){\vector(0,1){54}}

\put(-4,53){${\scriptstyle b_1}$}
\put(6,53){${\scriptstyle b_2}$}
\put(18,53){$. . .$}
\put(36,53){${\scriptstyle b_n}$}

\put(-4,-13){${\scriptstyle j_1}$}
\put(7,-13){${\scriptstyle j_2}$}
\put(18,-13){$. . .$}
\put(38,-13){${\scriptstyle j_n}$}

\put(52,39){${\scriptstyle a_1}$}
\put(53,28){${\scriptstyle a_2}$}
\put(51,13){$\vdots$}
\put(54,-2){${\scriptstyle a_n}$}

\put(-16,39){${\scriptstyle i_1}$}
\put(-16,28){${\scriptstyle i_2}$}
\put(-15,13){$\vdots$}
\put(-16,-2){${\scriptstyle i_n}$}

\put(125,28){
\put(-10,0){\vector(1,0){20}}
\put(0,-10){\vector(0,1){20}}
\put(-17,-3.5){$i$}\put(12.5,-3.5){$a$}
\put(-2.4,13){$b$}\put(-2.3,-19){$j$}
\put(30,-4){$=\, \hat{R}^{ab}_{ij}$\, in (\ref{vt}).}
}
\end{picture}
\label{ask2}
\end{equation}

In what follows, the $n^2$-fold tensor product contained in 
(\ref{ask2}) will always be arranged according to the order
specified in Example \ref{ex:n=3}.

\begin{proposition}[Embedding of $X_\alpha(z)$ into 
$q=0$ layer to layer transfer matrix]\label{pr:sdkf}
For any $r \in \Z_{\ge 0}$, the 
${\mathbb T}(z)^{0, \ldots ,0,r}_{0, \ldots ,0,r}$ 
is expanded in terms of the $n$-$\mathrm{TAZRP}$ 
operator $X_\alpha(z)\; \mathrm{(\ref{Xz})}$ as  
\begin{equation}\label{S=X}
\begin{split}&{\mathbb T}(z)^{0, \ldots ,0,r}_{0, \ldots ,0,r}\\
&=z^{-r}\sum_{\alpha \in (\Z_{\ge 0})^n}X_\alpha(z) \otimes
\overbrace{({\bf a}^+)^{\alpha_{\ge n}}({\bf a}^-)^r \otimes 
\cdots \otimes ({\bf a}^+)^{\alpha_{\ge 1}}({\bf a}^-)^r}^{\rm diagonal}{\otimes}\overbrace{({\bf a}^+)^r{\otimes}{\cdots}
{\otimes}({\bf a}^+)^r}^{n-1}\otimes{\bf 1}^{\otimes N},
\end{split}
\end{equation}
where $N=(n-1)(n-2)/2$ and 
``diagonal" signifies the components corresponding 
to the vertices on the NE-SW diagonal of (\ref{ask2}).
\end{proposition}
\begin{proof}
We explain the proof along the $n=3$ case.
The general case is similar.
The vertex $\hat{R}^{ab}_{ij}$ (\ref{vt}) is $0$ unless 
the conservation law $a+b=i+j$ is satisfied and $a\ge j, b\le i$.
This property and the boundary condition for 
$\mathbb{T}(z)^{0,0,r}_{0,0,r}$ restrict the 
sum (\ref{ask2}) to the following configurations:
\begin{equation*}
\begin{picture}(300,105)(30,-32)
\setlength\unitlength{0.4mm}
\thinlines

\put(208,53){\color{red}$r$}
\put(188,53){$0$}
\put(168,53){$0$}

\put(152,33){$0$}
\put(152,13){$0$}
\put(152.5,-7){\color{red} $r$}

\put(229,33){$a_1(=\!\beta_3)$}
\put(229,13){$a_2$}
\put(229,-7){$a_3$}

\put(212,23){$ \beta_3$} 
\put(196,18){$ \beta_2$} 
\put(192,2){$ \beta_2$} 
\put(177,-2){$ \beta_1$} 

\put(142,-23){$(\beta_1\!=)\,j_1$}
\put(188,-23){$j_2$}
\put(208,-23){$j_3$}

\put(60,15){${\displaystyle
\mathbb{T}(z)^{0,0,r}_{0,0,r} = z^{-r}\sum_{{\bf a}, \,{\bf j}}
z^{|{\bf j}|}}$}

\put(150,-5){
\put(10,0){

\multiput(29,40)(-2,0){16}{.}
\multiput(29,42)(0,2){5}{.}
\put(30.3,52){\vector(0,1){3}}

\multiput(9,20)(0,2){16}{.}
\multiput(9,20)(-2,0){6}{.}

\put(10,50){\vector(0,1){5}}

\put(10,-10){\line(0,1){10}}
\put(10,0){\line(1,0){20}}
\put(30,-10){\line(0,1){10}}
\put(30,0){\line(0,1){20}}
\put(50,20){\line(0,1){20}}
\put(30,20){\line(1,0){20}}
\put(50,40){\vector(1,0){15}}
\put(50,20){\vector(1,0){15}}
\put(30,0){\vector(1,0){35}}
\put(50,-10){\line(0,1){30}}

\put(30,40){\color{red}\line(1,0){20}} 
\put(50,40){\color{red}\vector(0,1){15}}
\put(10,20){\color{red}\line(1,0){20}} 
\put(30,20){\color{red}\line(0,1){20}}
\put(-1,0){\color{red}\line(1,0){11}}
\put(10,0){\color{red}\line(0,1){20}}
}
}
\end{picture}
\end{equation*}
Here edges on the red line are frozen to $r$ and  
those on the dotted lines are so to $0$. 
The black edges are yet to be summed over but 
$\beta_1\ge \beta_2 \ge \beta_3$ must be satisfied.
Thus introducing 
$\alpha=(\alpha_1,\alpha_2,\alpha_3)$ by 
$\beta_j = \alpha_{\ge j}$, we see that 
the SE quadrant yields $X_{\alpha}(z)$. 
(See Example \ref{ex:n=3} for the configurations with the 
same boundary condition.)
The ``diagonal" part in (\ref{S=X}) is deduced from 
$\hat{R}^{\beta_j r}_{r \beta_j}
=({\bf a}^+)^{\beta_j}({\bf a}^-)^r$.
The ``$n\!-\!1$" part comes from 
$\hat{R}^{r,0}_{0,r}=({\bf a}^+)^r$.
The NW quadrant gives ${\bf 1}^{\otimes N}$.
As for the spectral parameter,
$X_{\alpha}(z)$ needs $z^{|{\bf a}|}$ from each summand,
but this is equal to $z^{|{\bf j}|}$ by the conservation law.
\end{proof}
\begin{example}
For $n=2$ one has (red lines denotes $r$)
\begin{equation*}
\begin{picture}(500,117)(-5,-12)
\thinlines
\put(50,80){
\put(-40,0){${\mathbb T}(z)^{0,r}_{0,r}
={\displaystyle \sum_{\alpha_1,\alpha_2}}{\displaystyle \sum_{j}}$}
\mbox{\begin{picture}(37,18)(-5,5)
\put(50,0){
\put(6,14){\vector(0,1){10}}
\put(18,14){\color{red}\vector(0,1){10}}
\put(-2,14){\line(1,0){8}}
\put(6,14){\color{red}\line(1,0){12}}
\put(18,14){\vector(1,0){11}}
\put(6,0){\color{red}\line(0,1){14}}
\put(18,0){\line(0,1){14}}
\put(-2,0){\color{red}\line(1,0){8}}
\put(6,0){\line(1,0){12}}
\put(18,0){\vector(1,0){11}}
\put(6,-8){\line(0,1){8}}
\put(18,-8){\line(0,1){8}}

\put(17,-14){\scriptsize{$j$}}
\put(30,-1){\scriptsize{$j+\alpha_1$}}
\put(30,13){\scriptsize{$\alpha_2$}}
}
\end{picture}
\put(-50,-40){$=z^{-r}{\displaystyle \sum_{\alpha_1,\alpha_2}}(z^{\alpha_1+\alpha_2}{\displaystyle \sum_{j}}z^j({\bf a}^+)^j{\bf k}^{\alpha_1}({\bf a}^-)^{\alpha_2}){\otimes}({\bf a}^+)^{\alpha_2}({\bf a}^-)^r{\otimes}({\bf a}^+)^{\alpha_1+\alpha_2}({\bf a}^-)^r{\otimes}({\bf a}^+)^r$}
\put(-50,-80){$=z^{-r}{\displaystyle \sum_{\alpha_1,\alpha_2}}X_{\alpha_1,\alpha_2}(z){\otimes}({\bf a}^+)^{\alpha_2}({\bf a}^-)^r{\otimes}({\bf a}^+)^{\alpha_1+\alpha_2}({\bf a}^-)^r{\otimes}({\bf a}^+)^r$,}

}
}
\end{picture}
\end{equation*}
where the last step is due to Example \ref{ex:n=2}.
\end{example}

\subsection{\mathversion{bold}Bilinear relations at $q=0$}

A simple inspection of the proof of Proposition \ref{pr:sdkf}
shows that 
$\mathbb{T}(z)^{0,\ldots, 0,s}_{0,\ldots, 0, r}=0$
unless $r\ge s$.
Therefore specialization of (\ref{xySS}) to $q=0$ leads to
\begin{corollary}\label{co:dery}
For any $r \in \Z_{\ge 0}$, 
the $q=0$ layer to layer transfer matrices obey the bilinear relation
\begin{equation*}
\sum_{\substack{r_1,r_2\\r_1+r_2=r}}x^{2r_1}y^{2r_2}
{\mathbb T}(x)^{0,\ldots, 0, r_1}_{0,\ldots, 0, r_1}
{\mathbb T}(y)^{0,\ldots, 0, r_2}_{0,\ldots, 0, r_2}=(x \leftrightarrow y).
\end{equation*}
\end{corollary}

It is natural to substitute Proposition \ref{pr:sdkf} into 
Corollary \ref{co:dery} and seek the consequence on the 
TAZRP operators.
The result is given by 
\begin{proposition}[Bilinear identities of TAZRP operators]
For any $\alpha_,\beta \in (\Z_{\ge 0})^n$,
\begin{equation}\label{xXX}
x^{|\beta|}{\displaystyle \sum_{(\gamma,\delta)\ge(\alpha,\beta)}}X_{\gamma}(x)X_{\delta}(y)=(x \leftrightarrow y),
\end{equation}
where $(\gamma,\delta)\ge(\alpha,\beta)$ 
has been defined around $\mathrm{(\ref{smta})}$.
\end{proposition}
\begin{proof}
Fix $\alpha_,\beta \in (\Z_{\ge 0})^n$.
Set $r= |\beta |$ and consider 
Corollary \ref{co:dery}  with respect to this $r$:
\begin{align}\label{hzks}
x^{2r}{\mathbb T}(x)^{0,\ldots, 0, r}_{0,\ldots,0, r}
{\mathbb T}(y)^{0,\ldots, 0, 0}_{0,\ldots, 0, 0} + 
\sum_{\substack{r_1< r, r_2>0 \\r_1+r_2=r}}x^{2r_1}y^{2r_2}
{\mathbb T}(x)^{0,\ldots, 0, r_1}_{0,\ldots, 0, r_1}
{\mathbb T}(y)^{0,\ldots, 0, r_2}_{0,\ldots, 0, r_2}=(x \leftrightarrow y).
\end{align}
The first term here is calculated by applying 
Proposition \ref{pr:sdkf} as  
\begin{equation}\label{akci}
x^r\sum_{\gamma,\delta}X_{\gamma}(x)X_{\delta}(y)\otimes
\Bigl(\otimes_{n \ge i \ge 1}
({\bf a}^+)^{\gamma_{\ge i}}({\bf a}^-)^r({\bf a}^+)^{\delta_{\ge i}}\Bigr)
\otimes\overbrace{({\bf a}^+)^r \otimes\cdots
\otimes({\bf a}^+)^r}^{n-1}\otimes{\bf 1}^{\otimes N},
\end{equation}
where the sum is over 
$\gamma, \delta \in (\Z_{\ge 0})^n$ and 
$\otimes_{n \ge i \ge 1}$ is arranged from $i=n$ in the left 
to $i=1$ in the right.
As this term exemplifies, the rightmost $n-1+N$ components only 
yield a common overall factor free from $x,y$ for all the terms appearing in (\ref{hzks}).
Therefore we omit them in the sequel and call the resulting rightmost part 
$\otimes_{n \ge i \ge 1}(\cdots)$ as the {\em diagonal component}.

From (\ref{rna}) the product 
$({\bf a}^+)^{f}({\bf a}^-)^r({\bf a}^+)^{g}$ is reduced to 
$({\bf a}^+)^f({\bf a}^-)^{r-g}$ if $r \ge g$ and 
$({\bf a}^+)^{f+g-r}$ if $r \le g$.
According to this alternative, 
(\ref{akci}) is expanded, in view of 
$\delta_{\ge n} \le \cdots \le \delta_{\ge 1}$, as
\begin{align}\label{anm}
x^r\sum_{m=0}^n
\sum_{\substack{\gamma,\delta\\ \delta_{\ge m+1}<r\le\delta_{\ge m}}}
\!\!\!\!\!\!\!\!\! X_{\gamma}(x)X_{\delta}(y)\otimes
\begin{bmatrix} 
\gamma_{\ge n}\;\;  \ldots  \;\;\gamma_{\ge m+1} &
\! \! \! \! \! \gamma_{\ge m}\!+\!\delta_{\ge m} \!-\!r  \;\ldots\;
 \gamma_{\ge 1}\!+\!\delta_{\ge 1} \!-\! r\\ 
r\!-\!\delta_{\ge n} \;\ldots \; r\!-\!\delta_{\ge m+1} 
&0  \quad \quad  \quad\ldots \quad  \quad  \quad 0
\end{bmatrix},
\end{align}
where $\delta_{\ge 0} = \infty, \delta_{\ge n+1} = -1$ and 
$\left[ f_n, ... , f_1 \atop g_n, ... , g_1\right]=
\otimes_{n \ge i \ge 1}({\bf a}^+)^{f_i}({\bf a}^-)^{g_i}$.
We will only encounter the situation $g_n \ge \cdots \ge g_1$ 
in the sequel. 
Notice that the second term in (\ref{hzks}) and 
the $m=0$ term in (\ref{anm})  have the diagonal component 
of the form $\left[ f_n, ... , f_1 \atop g_n, ... , g_1\right]$ with 
$\forall g_i \ge 1$ only, whereas 
the $m \in [1,n]$ terms in (\ref{anm}) contain
those with $g_n \ge \cdots \ge g_{m+1} > g_m=\cdots = g_1=0$.
At this point we utilize the fact that 
$\{({\bf a}^+)^f({\bf a}^-)^g| f,g  \in \Z_{\ge0}\}$ forms a basis 
of $\mathscr{A}_0$.
See the remark after (\ref{rna}).
Thus (\ref{hzks}) leads to the separated identities for each $m \in [1,n]$:
\begin{align}\label{anch}
x^r\!\!\!\!\!\!\!\!\!\!\!
\sum_{\substack{\gamma,\delta\\ \delta_{\ge m+1}<r\le\delta_{\ge m}}}
\!\!\!\!\!\!\!\!\! X_{\gamma}(x)X_{\delta}(y)\otimes
\begin{bmatrix} 
\gamma_{\ge n}\;\;  \ldots  \;\;\gamma_{\ge m+1} &
\! \! \! \! \! \gamma_{\ge m}\!+\!\delta_{\ge m} \!-\!r  \;\ldots\;
 \gamma_{\ge 1}\!+\!\delta_{\ge 1} \!-\! r\\ 
r\!-\!\delta_{\ge n} \;\ldots \; r\!-\!\delta_{\ge m+1} 
&0  \quad \quad  \quad\ldots \quad  \quad  \quad 0
\end{bmatrix} = (x\leftrightarrow y).
\end{align}
One can further separate (\ref{anch}) 
according to the diagonal component.

\vspace{0.2cm}
(i)  Case $r=|\beta|>0$. 
There is a unique $l \in [1,n]$ 
such that $\beta=(0,\ldots, 0, \beta_l,\dots, \beta_n)$ and $\beta_l>0$.
Pick the terms in the left hand side of 
$(\ref{anch})|_{m=l}$ whose diagonal component is 
\begin{align}\label{elvk}
\begin{bmatrix} 
\alpha_{\ge n}\;\;  \ldots  \;\;\alpha_{\ge l+1} &
\! \! \! \! \! \alpha_{\ge l}\!+\!\beta_{\ge l} \!-\!r  \;\ldots\;
 \alpha_{\ge 1}\!+\!\beta_{\ge 1} \!-\! r\\ 
r\!-\!\beta_{\ge n} \;\ldots \; r\!-\!\beta_{\ge l+1} 
&0  \quad \quad  \quad\ldots \quad  \quad  \quad 0
\end{bmatrix}.
\end{align}
It amounts to imposing the following conditions on the sum 
$x^r \sum_{\gamma,\delta}X_{\gamma}(x)X_{\delta}(y)$:
\begin{align*}
\gamma_j+\delta_j = \alpha_j\;\;(j \in [1,l-1]),\quad\;
\gamma_l+\delta_l=\alpha_l+\beta_l,\;
\delta_l \ge \beta_l,\quad\;
(\gamma_j,\delta_j)=(\alpha_j, \beta_j)\;\; (j \in [l+1,n]).
\end{align*}
By (\ref{smta}) this is nothing but  
$(\gamma, \delta) \ge (\alpha, \beta)$,
proving (\ref{xXX}).

\vspace{0.2cm}
(ii)  Case $r=|\beta|=0$. 
We have $\beta={\bf 0} \in (\Z_{\ge 0})^n$.
Pick the terms in the left hand side of 
$(\ref{anch})|_{m=n}$ whose diagonal component is 
$(\ref{elvk})|_{\beta={\bf 0}, r=0}$.
It leads to the sum 
$x^r \sum_{\gamma,\delta}X_{\gamma}(x)X_{\delta}(y)$
with the condition
$\gamma_j+\delta_j = \alpha_j$ for $j \in [1,n]$.
Again by (\ref{smta}) 
this is equivalent to $(\gamma, \delta) \ge (\alpha, \beta={\bf 0})$.
\end{proof}

\vspace{0.2cm}
{\it Proof of Theorem \ref{main}}. 
Differentiate (\ref{xXX}) with respect to $x$ and set $x=y=z$. 
The result reads
\begin{align}\label{bXX}
z\sum_{(\gamma,\delta)\ge(\alpha,\beta)}
(X'_\gamma(z)X_\delta(z)-X_\gamma(z)X'_\delta(z))
=-|\beta|\sum_{(\gamma,\delta)\ge(\alpha,\beta)}X_\gamma(z)X_\delta(z).
\end{align}
Let $\{(\alpha^{(i)},\beta^{(i)})\}$ be the set of minimal elements 
in $\{(\gamma,\delta)|(\gamma,\delta)>(\alpha,\beta)\}$ 
with respect to $\ge$.
By the definition 
$(\alpha^{(i)},\beta^{(i)})>(\alpha,\beta)$ 
and $|\beta^{(i)}|=|\beta|+1$ are valid. 
Here is an $n=2$ example for 
$(\alpha,\beta)=((1,2),(0,0))$, where the elements in 
$\{(\gamma,\delta)|(\gamma,\delta)>(\alpha,\beta)\}$ 
are partially ordered as follows:
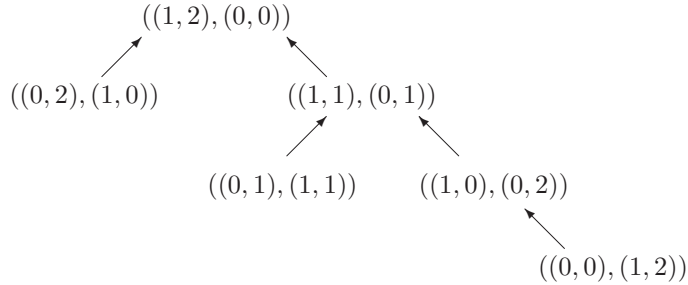
\begin{equation*}
\begin{picture}(150,115)(-30,0)
\put(0,-20){
\put(-20,120){$((1,2),(0,0))$}
\put(-35,100){\vector(1,1){15}}
\put(-70,90){$((0,2),(1,0))$}
\put(50,100){\vector(-1,1){15}}
\put(35,90){$((1,1),(0,1))$}
\put(35,70){\vector(1,1){15}}
\put(5,56){$((0,1),(1,1))$}
\put(100,70){\vector(-1,1){15}}
\put(85,56){$((1,0),(0,2))$}
\put(140,35){\vector(-1,1){15}}
\put(130,25){$((0,0),(1,2))$}
}
\end{picture}
\end{equation*}
Here $(\gamma,\delta)\rightarrow(\alpha,\beta)$ denotes $(\gamma,\delta)\ge(\alpha,\beta)$. The minimal elements
are $((0,2),(1,0))$ and $((1,1),(0,1))$.
We consider the relation
$(\ref{bXX})|_{(\alpha,\beta)\to (\alpha^{(i)},\beta^{(i)})}$
and subtract it from $(\ref{bXX})$
for all the minimal elements $(\alpha^{(i)},\beta^{(i)})$.
In the process 
each $(\gamma,\delta)$-term in (\ref{bXX}) except $(\alpha, \beta)$ 
is subtracted exactly once because 
$\{(\alpha,\beta)|(\gamma,\delta)\ge(\alpha,\beta)\}$ 
with fixed $(\gamma,\delta)$ is a totally ordered set with respect to $\ge$ and therefore, any $(\gamma,\delta)$ such that 
$(\gamma,\delta)>(\alpha,\beta)$ has the unique minimal element 
$(\alpha^{(i)},\beta^{(i)})$ 
satisfying $(\gamma,\delta)\ge(\alpha^{(i)},\beta^{(i)})$.
Thus the subtraction yields
\begin{equation}\label{ngjj}
\begin{split}
z(X'_\alpha(z)X_\beta(z)-X_\alpha(z)X'_\beta(z))
&=\sum_{(\gamma,\delta)>(\alpha,\beta)}
X_\gamma(z)X_\delta(z)-|\beta|X_\alpha(z)X_\beta(z)\\
&= \sum_{\gamma,\delta}
h^{\alpha,\beta}_{\gamma,\delta}X_\gamma(z)X_\delta(z)
\end{split}
\end{equation}
owing to (\ref{hrk2}).
The proof is finished by setting $z=1$ in (\ref{ngjj}).
\qed

\vspace{0.2cm}
One may undo the specialization to $z=1$ in (\ref{XX})
and consider the $z$-dependent matrix product
$\mathrm{Tr}(X_{ \sigma_1}(z)\cdots X_{ \sigma_L}(z))$.
The Baxterized hat relation (\ref{ngjj}) tells that it still satisfies the 
steady state probability condition.
However, by the uniqueness of the steady state,
it coincides with the $z=1$ case up to an overall power of $z$.

\section*{Acknowledgments}
This work is supported by 
Grants-in-Aid for Scientific Research No.~15K04892,
No.~15K13429 and No.~23340007 from JSPS.

\end{document}